\documentclass{article}

\usepackage{PRIMEarxiv}

\usepackage[utf8]{inputenc} 
\usepackage[T1]{fontenc}    
\usepackage{hyperref}       
\usepackage{url}            
\usepackage{booktabs}       
\usepackage{amsfonts}       
\usepackage{nicefrac}       
\usepackage{microtype}      
\usepackage{lipsum}
\usepackage{fancyhdr}       
\usepackage{graphicx}       
\graphicspath{{media/}}     
\usepackage{balance} 

\usepackage{natbib}
\usepackage{makecell}
\usepackage{etoolbox}                     
\usepackage{hyperref}     
\usepackage{amsthm}
\usepackage{amsmath}
\usepackage{amssymb}
\usepackage[nameinlink,noabbrev,capitalise]{cleveref}

\usepackage{subcaption}

\theoremstyle{remark}
\newtheorem*{remark}{Remark}

\newcommand{\AWNS}{\textsc{AWNS}}
\newcommand{\AWNSD}{\textsc{AWNS}\mbox{-}\textsc{d}}
\newcommand{\AWNSDC}{\textsc{AWNS}\mbox{-}\textsc{d}\mbox{-}\textsc{c}}
\newcommand{\AWNSR}{\textsc{AWNS}\mbox{-}\ensuremath{\rho}}
\newcommand{\AWNSRC}{\textsc{AWNS}\mbox{-}\ensuremath{\rho}\mbox{-}\textsc{c}}
\newcommand{\AWNSNW}{\textsc{AWNS}\mbox{-}\textsc{NW}}

\newcommand{\appsymb}{$\bigstar$}
\newcommand{\appref}[1]{{\hyperref[proof:#1]{\appsymb}}}
\newcommand{\apprefX}[1]{{\hyperref[#1]{\appsymb}}}

\newcommand{\Rof}[1]{\ensuremath{R_{[#1]}}}

\usepackage{xspace}
\usepackage{xcolor}
\definecolor{tuerkis}{RGB}{64, 224, 208}
\usepackage{etoolbox}
\usepackage{marginnote}
\marginparwidth=\dimexpr \marginparwidth + 0.5cm\relax
\usepackage[textsize=scriptsize]{todonotes} 

\newtheorem{definition}{Definition}
\newtheorem{example}{Example}
\newtheorem{proposition}{Proposition}
\newtheorem{theorem}{Theorem}
\newtheorem{lemma}{Lemma}
\newtheorem{corollary}{Corollary}
\newtheorem{claim}{Claim}

\title{Adjusted Winner: from Splitting to Selling
}

\author{
  Robert Bredereck, Bin Sun \\
  Institut für Informatik \\
  TU Clausthal \\
  Clausthal-Zellerfeld, Germany\\
  \texttt{\{robert.bredereck,bin.sun\}@tu-clausthal.de} \\
   \And
  Eyal Briman, Nimrod Talmon \\
  Department of Industrial Engineering and Management \\
  Ben-Gurion University of the Negev \\
  Beer Sheva, Israel\\
  \texttt{briman@post.bgu.ac.il,talmonn@bgu.ac.il} \\
}

\newcommand{\BibTeX}{\rm B\kern-.05em{\sc i\kern-.025em b}\kern-.08em\TeX}

\newcommand{\mytodo}[2]{\todo[size=\tiny, color=#1!50!white]{#2}\xspace}
\newcommand{\myrevtodo}[2]{{%
		\let\marginpar\marginnote
		\reversemarginpar
		\renewcommand{\baselinestretch}{0.8}%
		\todo[size=\tiny, color=#1!50!white]{#2}\xspace}}
\newcommand{\myinlinetodo}[2]{\todo[size=\small, color=#1!50!white, inline, caption={}]{#2}\xspace}
\newcommand{\registerAuthor}[3]{%
	\expandafter\renewcommand\csname #2com\endcsname[1]{\mytodo{#3}{\textsc{#2}:
	##1}}%
	\expandafter\renewcommand\csname
	#2revcom\endcsname[1]{\myrevtodo{#3}{\textsc{#2}: ##1}}%
	\expandafter\renewcommand\csname
	#2inline\endcsname[1]{\myinlinetodo{#3}{\textsc{#2}: ##1}}%
	\expandafter\renewcommand\csname
	#2inlineLater\endcsname[1]{\lv{\myinlinetodo{#3}{\textsc{#2}: ##1}}}%
}
\registerAuthor{Robert}{rb}{yellow}
\registerAuthor{Eyal}{eb}{green!50!gray}
\registerAuthor{Bin}{bs}{tuerkis}
\registerAuthor{Nimrod}{nt}{orange!50!white}

\begin{document}


\pagestyle{fancy}
\fancyhead{}

\maketitle

\begin{abstract}
The \emph{Adjusted Winner} (AW) method is a fundamental procedure for the fair division of \emph{indivisible} resources between two agents. However, its reliance on splitting resources can lead to practical complications. To address this limitation, we propose an extension of AW that allows the sale of selected resources under a budget constraint, with the proceeds subsequently redistributed, thereby aiming for allocations that remain as equitable as possible. Alongside developing this extended framework, we provide an axiomatic analysis that examines how equitability and envy-freeness are modified in our setting. We then formally define the resulting combinatorial problems, establish their computational complexity, and design a fully polynomial-time approximation scheme (FPTAS) to mitigate their inherent intractability. Finally, we complement our theoretical results with computer-based simulations.
\end{abstract}



\section{Introduction}

The \emph{Adjusted Winner} (AW) method~\cite{brams1996fair} is a foundational protocol for fair division between two agents. It assumes that each agent has a confidential additive utility function $u_1, u_2: \mathcal{P}(R) \to \mathbb{N}_0$ over a finite set of indivisible resources $R$, and is facilitated by a neutral mediator. The agents' utility functions are defined on a common cardinal scale,
i.e., $\sum_{r\in R} u_1(r) = \sum_{r\in R} u_2(r)$, so that their utilities are directly comparable. The allocation proceeds in two phases:

\begin{description}
  \item [Phase 1 (Initial Allocation):] Each resource \( r \in R \) is allocated to the agent who values it more: if \( u_1(r) > u_2(r) \),
  then \( r \) is assigned to Agent~1; otherwise, it is assigned to Agent~2. This results in a tentative allocation \((S_1, S_2)\).
  \item [Phase 2 (Adjustment):]  
  To equalize the total utilities of the two agents, the algorithm identifies
  the agent who currently has the higher total utility (the \emph{advantaged}
  agent~$a$) and considers transferring resources from agent~$a$ to the other agent~$b$.  
  The resources owned by the advantaged agent are sorted in non-decreasing
  order of the ratio \(\frac{u_a(r)}{u_b(r)}\).  
  Resources are then transferred in this order until the agents’ total utilities
  become equal, \(u_1(S_1)=u_2(S_2)\).  
  If exact equality cannot be achieved by transferring whole resources, the
  last transferred resource is divided fractionally between the agents; this
  fractional transfer is referred to as a \emph{split}.
\end{description} 

Despite its prominence,\footnote{\url{https://pages.nyu.edu/adjustedwinner/}}
 some theoretical properties and computational aspects of AW remain relatively under-explored~\cite{aziz2015adjusted,salame2005some,massoud2000fair,pacuit2011towards,dupuis2009empirical,brams1997fair}. The method is especially notable for guaranteeing \emph{equitable}, \emph{envy-free}, and \emph{Pareto-optimal} allocations under additive valuations. These strong fairness properties, however, crucially depend on the ability to divide one (arbitrary) indivisible resource. Indeed, the possibility of splitting a contested resource is fundamental to AW's guarantees, which can render the method less applicable in settings where such divisions are infeasible or undesirable. 


Our study addresses the practical infeasibility of resource splitting by
permitting the sale of selected resources, subject to a global budget and
resource-specific costs that capture the logistical or emotional burdens
involved. We deliberately retain the Adjusted Winner (AW) framework
because it is both simple and transparent, allowing parties to clearly
understand how the allocation is determined.  Yet the reliance of AW on fractional division is often unrealistic.
For instance, a shared rented apartment or a sentimental family heirloom cannot be literally divided, and even selling such resources may involve substantial emotional or procedural costs.
Our extension preserves the intuitive and user-friendly nature of AW while replacing fractional splits with resource sales, maintaining its transparency while accommodating real-world constraints.


\paragraph{Optimizing Equitability under Indivisibility Constraints.}
Instead of maintaining a formal guarantee of equitability at the cost of ignoring the indivisibility of resources, as in the classic AW method, we explicitly model indivisibility and aim to optimize equitability within feasible bounds.
While exact equitability may still be attainable in special cases, we do not enforce it when it would require impractical or infeasible splits. The framework (which can be seen as an extension of AW) permits the sale of multiple resources under a mediator-defined global budget, allowing equitability to be optimized subject to indivisibility, selling costs, and potential revenue from sold resources.

\paragraph{Welfare from Selling.}
Allowing the sale of resources naturally requires integrating both the
agents' utilities for allocated resources and the revenue generated from
selling resources. Accordingly, we define each agent's welfare as the sum
of the utility they obtain from their allocated resources and the share of
the revenue they receive from sold resources. Once a plan specifying which
resources are allocated and which are sold has been determined, the total
revenue from the sold resources is divided between the two agents so as to
make their resulting welfares as balanced as possible. 

Although agents cannot resell the resources they receive within our
setting, it is reasonable to expect that, if a resale were to occur later
on the open market, the obtainable revenue would not exceed the higher valuation among the two agents.
 At the same time, in many practical settings such as urgent liquidations, private exchanges, or emotionally constrained sales, market transactions tend to under-capture value. Therefore, we assume that the revenue obtainable from selling any resource does not exceed the higher utility that either agent assigns to it.

\paragraph{Cost of Selling.}
In this framework, selling a resource comes with its own cost, which might reflect the effort, time, or extra fees required to complete a sale. There is an overall limit (budget) on how much can be spent on selling. Importantly, these selling costs are of a completely different nature from the benefit or satisfaction the agents receive from the resources themselves.  For example, when dividing up property during a divorce, selling a car or a house often involves a lot of paperwork, legal procedures, and sometimes additional expenses.

A particularly relevant special case arises when costs are restricted to
\(\{0, \infty\}\) or to unary values. The former allows one to designate certain
resources as unsellable (by assigning an infinite cost) or freely sellable
(with zero cost), while the latter effectively limits the number of resources
that can be sold. These special cases make the framework flexible enough
to capture both continuous and discrete forms of selling constraints.

\paragraph{Positive Welfare.}
Since in our setting all utilities and prices are nonnegative, it is straightforward to see that both agents obtaining zero welfare constitutes a trivial yet uninformative form of fairness. To exclude such degenerate cases, we require that any feasible plan must yield strictly positive welfare for both agents.

\paragraph{Our contributions.}
We introduce the \emph{Dispute Settlement with Indivisible Resources and Sale (DSIRS)} framework, which provides a principled foundation for resolving disputes in contexts where resources cannot be divided, either for practical or emotional reasons. Within this framework, we formalize a family of Adjusted Winner without Splitting (AWNS) combinatorial optimization problems and analyze both their computational complexity and fairness properties. In particular, we establish an FPTAS for the AWNS-$\rho$ problem, a parameterized variant in which $\rho$ measures the ratio between the agents’ achieved welfare levels. Finally, we complement our theoretical results and modeling by computer-based simulations.

\section{Related Work}
We briefly survey prior work related to fairness, indivisibility, and conflict resolution in fair division problems.


\paragraph{Control in Allocation and Elections.}
Modifying the availability of resources to influence outcomes is conceptually related to control in voting, where agents alter the set of candidates or voters to achieve their preferred result~\cite{bartholdi1992hard,faliszewski2016control,aziz2016control}. In fair division, similar ideas arise when resources are pre-selected, withheld, or reweighted to affect final allocations~\cite{aziz2016control}. Our approach can be interpreted as a constrained form of resource selling aimed at preventing splits rather than optimizing strategic advantage.

\paragraph{Fair Division as Conflict Resolution.}
Fair division mechanisms have long been used as tools for conflict resolution, particularly in situations where indivisibilities lead to impasses~\cite{brams1996fair,bouveret2016characterizing,kilgour2024two}. In such cases, modifying the problem instance---e.g., by deleting or monetizing contentious resources---can facilitate compromise. Our work formalizes this intervention-based approach within the AW framework, providing theoretical and computational guarantees for fairness- and efficiency-oriented resource selection strategies.

\paragraph{Resource Sale (Donation) and Envy-Freeness.}
Bilò et al.~\cite{bilo2024achieving} consider achieving envy-freeness by selling resources and redistributing proceeds. Their work focuses on general allocation frameworks and provides a PTAS for two agents. Our model differs in that it retains the AW structure and leverages resource sale not for reallocation per se, but to satisfy constraints that preclude splits, thereby preserving properties of the original AW method under additional feasibility constraints. \citet{ChaudhuryKMS21} study the existence of EFX allocations for agents with general valuations and show that donating a small number of non-envied resources to charity, which can be seen as a degenerate case of a sale with zero payment, suffices to ensure the existence of such allocations.

\paragraph{Fairness Guarantees with Indivisible Resources.}  
Since we consider all resources to be truly indivisible---unlike in the Adjusted Winner method---we cannot guarantee equitability, envy-freeness and Pareto optimality.
This limitation is not unique to our approach: it is intrinsic to the setting. As emphasized in recent algorithmic surveys on discrete fair division~\cite{aziz2022algorithmic,amanatidis2022fair}, exact fairness guarantees such as envy-freeness and proportionality become infeasible once resources are indivisible, motivating relaxations such as envy-freeness up to one resource (EF1)~\cite{aleksandrovenvy,plaut2020almost}. Concurrently, empirical studies of family property and inheritance mediation report that disputants and neutrals prioritize enforceable whole-resource outcomes, confidentiality, and efficiency over fractional allocations or theoretically ideal fairness~\cite{folberg2009mediating}. Guided by this evidence, our model deliberately avoids unattainable guarantees and instead promotes a policy of minimal inequality: it ensures that the final allocation consists only of indivisible resources, and permits resource sales under predefined constraints to reduce disparity between the agents. \citet{ChakrabortyISZ21} propose another AW–style algorithm that respects indivisibility and guarantees weighted envy-freeness up to one item (EF1) for agents with arbitrary entitlements.

\paragraph{Divisible Revenue}
Halpern and Shah~\cite{HalpernS19} proposed achieving envy-freeness through external monetary subsidies, where a third party provides additional divisible resources to offset fairness violations. Relatedly, more recent work has studied conceptually equivalent mechanisms under different names, such as subsidies or dollars (money)~\cite{CaragiannisI21,BrustleDNSV20}. Our work is conceptually complementary: instead of adding external money, we restore balance by selling certain resources to generate divisible revenue within the system.

\section{Preliminaries}
In this section, we formally introduce the DSIRS framework for allocating indivisible resources with the option of selling some resources and examine the new challenges that arise under the no-split constraint. These considerations motivate a family of optimization problems that we study in the subsequent sections. 

Throughout this section and the remainder of the paper, proofs of selected results are omitted from the main text and provided in the full version.

\subsection{Framework (DSIRS)}
We begin with a commonly encountered scenario and develop a framework based on it. Two agents, Agent~1 and Agent~2, seek to allocate their indivisible assets.  We call this framework \textit{Dispute Settlement with Indivisible Resources and Sale (DSIRS)}. 
\begin{definition}[DSIRS]

An instance of \textit{Dispute Settlement with Indivisible Resources and Sale (DSIRS)} $\mathcal{I}=\langle R,u_1,u_2,p,c,B\rangle$ consists of:

\begin{itemize}
    \item a finite set of indivisible resources, $R = \{r_1, r_2, \dots, r_m\}$,
    \item an additive utility function $u_j: R \to \mathbb{N}_0$ for each agent $j \in \{1,2\}$, s.t. $\sum_{r\in R}u_1(r)=\sum_{r\in R}u_2(r)$, 

    \item an additive price (revenue) function $p: R \to \mathbb{N}_0$, where $p(r)$ denotes the revenue generated from selling $r$, s.t. $p(r) \leq \max\{ u_1(r), u_2(r)\}$, 
    
         \item an additive selling cost function $c: R \to \mathbb{N}_0$,
           \item a budget for total cost $B \in \mathbb{N}_0$. 
\end{itemize}

\begin{remark}
    By additivity, we denote $u_j(S) = \sum_{r \in S} u_j(r)$ for all $S \subseteq R$ and $j \in \{1,2\}$; $p(S) = \sum_{r \in S} p(r)$ for all $S \subseteq R$;  $c(S) = \sum_{r \in S} c(r)$ for all $S \subseteq R$. For convenience, when the indices of the resources are fixed and clear from the context, we denote $u_j(r_i)$, $p(r_i)$ and $c(r_i)$ by $u_{ij}$, $p_i$ and $c_i$, respectively.
\end{remark}
\end{definition}

Having defined an instance of DSIRS, we now specify what constitutes a solution. Every resource is required to be in exactly one of three mutually exclusive subsets: $S_0$, $S_1$, and $S_2$, where $S_0$ contains the resources sold to generate divisible funds, and $S_1$ and $S_2$ contain the resources allocated to agents 1 and 2, respectively. A tuple \(\mathcal{S(I)}=\langle S_0,S_1,S_2\rangle\) is called a \emph{plan (solution)} of an instance of DSIRS $\mathcal{I}=\langle R,u_1,u_2,p,c,B\rangle$. When the instance is clear from context, we write $\mathcal{S}$. The welfare functions for a plan $\mathcal{S} = \langle S_0, S_1, S_2 \rangle$
in instance $\mathcal{I}$ are defined as
\[
W_1(\mathcal{S}, \mathcal{I}) = u_1(S_1) + q \cdot p(S_0), \qquad
W_2(\mathcal{S}, \mathcal{I}) = u_2(S_2) + (1 - q) \cdot p(S_0),
\]
where the fraction $q \in [0,1]$ specifies how the revenue from the sold
resources is divided between the two agents. For any fixed
$S_0, S_1, S_2$, this fraction is uniquely determined by requiring the
two agents' welfares to be as balanced as possible. Formally,
\[
q = \max\!\left( 0,\ \min \!\left( 1,\ 
  \frac{p(S_0) - u_1(S_1) + u_2(S_2)}{2\,p(S_0)} \right) \right).
\]
Thus, $q$ is not a decision variable of the plan but a derived quantity
that ensures minimal welfare imbalance given the chosen sets.

\begin{remark}
\( W_i(\mathcal{S},\mathcal{I}) \) denotes the welfare that Agent~\( i \) obtains under plan~\( \mathcal{S} \) for instance~$\mathcal{I}$.   When the instance is clear but multiple plans are considered, we write $W_i(\mathcal{S})$ for the welfare of Agent~$i$ under plan~$\mathcal{S}$.  
When both the instance and plan are clear from context or when referring to the welfare of Agent~$i$ in general, we write \( W_i \) for brevity.
 Unless otherwise specified, the value of $q$ is determined by the choice of $S_0,S_1,S_2$ as described above. In certain special cases, we write $\mathcal{S(I)}$ as $\langle S_0, S_1, S_2, q'\rangle$, where the parameter $q$ is fixed to the given value $q'$.
 \end{remark}

A plan $\mathcal{S}$ is called \emph{feasible} if 
$W_1(\mathcal{S}), W_2(\mathcal{S}) > 0$ and 
$ c(S_0) \le B$.
That is, both agents obtain strictly positive welfare, and 
the total cost associated with selling resources does not exceed 
the budget~$B$. 
The quantity $c(S_0)$ is also referred to as the \emph{organizational cost}.
In the following, we restrict our attention to feasible plans.

\subsection{Fairness Analysis}
Here, we investigate some fairness properties of the DSIRS framework. We introduce two equitability thresholds: (a) welfare difference $d = |W_1 - W_2|$ and (b) welfare ratio $ \rho=\max\{\frac{W_1}{W_2},\frac{W_2}{W_1}\}$. The following proposition shows that, for any fixed partition $S_0, S_1, S_2$ of the resources, the same choice of $q$ simultaneously minimizes both objectives.

\begin{proposition}
Let $\mathcal{I}$ be an instance of DSIRS and let $S_0, S_1, S_2$ be a partition of the resources~$R$.
Define
\[
q_1 = \arg\min_{q' \in [0,1]} \left| u_1(S_1) - u_2(S_2) + (2q'-1)\cdot p(S_0) \right|,
\]
\[
q_2 = \arg\min_{q' \in [0,1]} \max \left\{\frac{W_1'}{W_2'},\frac{W_2'}{W_1'}
\right\},
\] where $\frac{W_1'}{W_2'}= \frac{u_1(S_1) + q'\cdot p(S_0)}{u_2(S_2) + (1-q')\cdot p(S_0)}$. Then $q_1 = q_2$.
\label{easysimult}
\end{proposition}


{

\begin{proof}
Since $S_0$, $S_1$, and $S_2$ are fixed for $\mathcal{I}$, the total welfare $u_1(S_1)+u_2(S_2)+p(S_0)$ is also fixed.  Observe that for any $A \geq B > 0$ and any $\epsilon \geq 0$, it holds that
\[
\frac{A+\epsilon}{B-\epsilon} \geq \frac{A}{B}.
\] Thus, minimizing $\left| u_1(S_1) - u_2(S_2) + (2q-1)\, p(S_0) \right|$ with respect to $q$ also minimizes $\max \left\{
    \frac{u_1(S_1) + q\cdot p(S_0)}{u_2(S_2) + (1-q)\cdot p(S_0)},
    \frac{u_2(S_2) + (1-q)\cdot p(S_0)}{u_1(S_1) + q\cdot p(S_0)}
\right\}$.
By symmetry, the same argument applies when $B \ge A > 0$ and any $\varepsilon \ge 0$.
  \end{proof}
}

\begin{corollary}
\label{easydrho}
For any plan $\mathcal{S}$,  $W_1(\mathcal{S})-W_2(\mathcal{S})=0$ if and only if $\frac{W_1(\mathcal{S})}{W_2(\mathcal{S})}=1$.
\end{corollary}

The above proposition shows that, for any \emph{fixed} partition $S_0, S_1, S_2$, the same choice of $q$
simultaneously minimizes both the welfare difference $d$ and the welfare ratio $\rho$. In contrast, when the
partition itself is part of the decision (i.e., $S_0, S_1, S_2$ are not given but chosen to optimize the objective),
the optimal partitions for minimizing $d$ and minimizing $\rho$ may differ. The following proposition formalizes this observation.

\begin{proposition}\label{conflicts3}
Given an instance of DSIRS, a plan $\mathcal{S}$ that minimizes $d$ may not minimize $\rho$, and vice versa.
\end{proposition}
{
\begin{proof}
Consider the instance with $R=\{a,b,c\}$. 
All unspecified values of $u_j(\cdot)$, $c(\cdot)$, and $p(\cdot)$ are zero. 
The nonzero parameters are 
$c(a)=B$, $c(b)=c(c)>B$, 
$u_1(a)=99$, $u_1(b)=1$, 
$u_2(a)=8$, $u_2(b)=90$, and $u_2(c)=2$.

\begin{itemize}
    \item $\mathcal{S}$ achieves the minimum $d=1$ (with $\rho=2$).
$$\mathcal{S} = \langle\{a\},\{b\},\{c\}\rangle,\,
\text{with}\, W_1(\mathcal{S}) =1,  W_2(\mathcal{S}) = 2.$$

\item $\mathcal{S'}$ achieves the minimum $\rho=\frac{99}{92}$ (with $d=7$):
$$\mathcal{S'} = \langle\emptyset,\{a\},\{b, c\}\rangle,\,
\text{with}\, W_1(\mathcal{S'}) = 99,  \,W_2(\mathcal{S'}) = 92.$$
\end{itemize}
Note that $a$ is the only resource that can be sold. If $a$ is sold, since $p(a) = 0$ and $u_1(c) = 0$, $b$ must go to Agent~1 and $c$ to Agent~2, otherwise at least one agent would have welfare $0$. We now enumerate all remaining feasible plans:
$$\mathcal{S}_1 = \langle\emptyset,\{a,c\},\{b\}\rangle,\,
\text{with}\, W_1(\mathcal{S}_1) = 99,  \,W_2(\mathcal{S}_1) = 90.$$
$$\mathcal{S}_2 = \langle\emptyset,\{b\},\{a,c\}\rangle,\,
\text{with}\, W_1(\mathcal{S}_2) = 1,  \,W_2(\mathcal{S}_2) = 10.$$
$$\mathcal{S}_3 = \langle\emptyset,\{b,c\},\{a\}\rangle,\,
\text{with}\, W_1(\mathcal{S}_3) = 1,  \,W_2(\mathcal{S}_3) = 8.$$
$$\mathcal{S}_4 = \langle\emptyset,\{a,b\},\{c\}\rangle,\,
\text{with}\, W_1(\mathcal{S}_4) = 100,  \,W_2(\mathcal{S}_4) = 2.$$ \end{proof}
}

 We also consider the widely studied concept of envy-freeness as a fairness criterion \cite{brandt2016handbook}, and provide a definition tailored to our framework. However, as we will demonstrate shortly, such a plan does not always exist in our setting. Furthermore, the pursuit of envy-freeness may be incompatible with our earlier objective of minimizing equitability thresholds such as $d$ or $\rho$.

\begin{definition}[Envy-freeness]\label{def:envy-freeness}
For a plan \(\mathcal{S}=\langle S_0,S_1,S_2\rangle\), 
define the (signed) envies:
\[
  \operatorname{envy}_1(\mathcal{S}) = W_1(\mathcal{S'}) - W_1(\mathcal{S}),
  \quad
  \operatorname{envy}_2(\mathcal{S}) = W_2(\mathcal{S'}) - W_2(\mathcal{S}),
\] where $\mathcal{S}'=\langle S_0,S_2,S_1,1-q\rangle$ and $q$ is determined by $\mathcal{S}$. The plan is \emph{envy-free} if
\[
  \operatorname{envy}_1(\mathcal{S}) \le 0
  \quad\text{and}\quad
  \operatorname{envy}_2(\mathcal{S}) \le 0.
\]
That is, neither agent would like to have the other’s bundle (with revenue share).
\end{definition}

\begin{proposition}\label{thm:ef-impossible}
Given an instance of DSIRS, an envy-free plan does not always exist.
\end{proposition}

{
\begin{proof}
Consider the instance with $R=\{a,b\}$ and budget $B=1$. 
All unspecified values of $u_j(\cdot)$, $c(\cdot)$, and $p(\cdot)$ are zero. 
The nonzero parameters are 
$u_1(a)=70$, $u_1(b)=30$, 
$u_2(a)=60$, $u_2(b)=40$, 
$p(a)=p(b)=20$, and 
$c(a)=c(b)=1$. Since selling both $\{a,b\}$ is forbidden, there are only two possible cases:
\begin{itemize}
    \item If $S_0=\varnothing$, the agent who doesn't get $a$ will envy the agent who gets $a$, since $u_1(a)>u_1(b)$ and $u_2(a)>u_2(b)$.
    \item If $S_0 = \{a\}$ or $S_0 = \{b\}$, the revenue is 20. Regardless of which agent receives the remaining resource, the other will envy him, since $\min\{u_1(a), u_1(b), u_2(a), u_2(b)\} \geq 20$. 
\end{itemize}

Hence, no envy-free plan exists for this instance.
     \end{proof}
}

\begin{proposition}
    \label{conflicts}
   Given an instance of DSIRS, a plan $\mathcal{S}$ with minimum \( d \) may not be envy-free, whereas an envy-free plan $\mathcal{S'}$ may not have minimum \( d \).
\end{proposition}

{
\begin{proof}
    Consider the following instance with $R=\{v,w,x,y,z\}$, $B=1$ and $\forall r\in R: c(r)=1$ and all unspecified values of $u_j(\cdot)$, $c(\cdot)$, and $p(\cdot)$ are zero:
\[
\begin{array}{c|ccccc}
       & v & w & x & y & z \\ \hline
u_1    & 22 & 56 & 16 & 1 & 25 \\
u_2    & 39 & 2  & 28 & 27 & 24 \\ \hline
p      & 2  & 5  & 4  & 10 & 5
\end{array}
\]

Consider the following two solutions:
\begin{itemize}
    \item $\mathcal{S}$ satisfies $|W_1(\mathcal{S})-W_2(\mathcal{S})|=0$ (not envy-free):\\
$\mathcal{S} = \langle\{y\},\{x, z\},\{v, w\}\rangle =\langle\{y\},\{x, z\},\{v, w\},0.5\rangle,\,
\text{with}\, W_1(\mathcal{S}) = W_2(\mathcal{S}) = 46, \,\mathrm{envy}_1(\mathcal{S}) = 37,\,\mathrm{envy}_2(\mathcal{S}) = 11.$

\item $\mathcal{S'}$ achieves no envy for either agent ($W_1(\mathcal{S'})-W_2(\mathcal{S'})=5$):\\
$\mathcal{S'} = \langle\{z\},\{w,y\},\{v, x\}\rangle = \langle\{z\},\{w,y\},\{v, x\},1\rangle
\text{ with } W_1(\mathcal{S'}) = 62,  \,W_2(\mathcal{S'}) = 67, \,\mathrm{envy}_1(\mathcal{S'}) = -24,\,\mathrm{envy}_2(\mathcal{S'}) = -33.$
\end{itemize}
   \end{proof}
}

By Corollary~\ref{easydrho} and Proposition~\ref{conflicts}, we immediately obtain the following corollary.
\begin{corollary}
    \label{conflicts2}
    Given an instance of DSIRS, a plan that has minimum \( \rho \) may not be envy-free, whereas an envy-free plan may not have minimum \( \rho \).
\end{corollary}

\subsection{Adjusted Winner without Splitting (\AWNS)}
The DSIRS framework closely parallels the Adjusted Winner (AW) algorithm. However, directly applying AW in this context may result in resource splits, violating the indivisibility constraint. Accordingly, for any subset $S_0 \subseteq R$ selected for sale, we consider only the partial execution of AW on the residual instance $R \setminus S_0$, halting the procedure immediately before a split would occur. This ensures that all resources in $R \setminus S_0$ are allocated integrally to one agent or the other. We refer to the result as the \emph{AW-derived plan}.


\begin{definition}[AW-derived plan]
Given an instance $\mathcal{I}$ and a subset $S_0 \subseteq R$ of resources to be sold, the \emph{AW sub-plan} $SP(S_0, \mathcal{I})$ is the allocation $(G_1, G_2)$ produced by running the Adjusted Winner (AW) algorithm on $R \setminus S_0$,
halting immediately in Phase~2 (a) before any split would occur or (b) when $|u_1(S_1)-u_2(S_2)| \le p(S_0)$ holds, since in that case the welfare gap between the agents can already be balanced by redistributing the sale revenue.
The \emph{AW-derived plan} $P_\textbf{AW}(S_0, \mathcal{I})$ is the triple $\langle S_0, S_1, S_2 \rangle$.
\end{definition}

The following example demonstrates how selling resources can eliminate the need for splitting.
\begin{example}

Consider a couple, Alex and Belle, who decide to divide their joint
possessions: a vintage Rolex watch $r_1$, four local art pieces
$r_2$---$r_5$, and a designer bag $r_6$.
Their valuations and potential sale revenues are listed below.
Under the classical Adjusted Winner (AW) procedure, the watch $r_1$
would need to be split to achieve equity, even after transferring all local art pieces from Alex to Belle (Alex~56, Belle~50).

\begin{center}
\begin{tabular}{lcccccc}
\toprule
 & $r_1$ & $r_2$ & $r_3$ & $r_4$ & $r_5$ & $r_6$ \\
\midrule
Alex & 56 & 11 & 11 & 11 & 11 & 0 \\
Belle & 50 & 10 & 10 & 10 & 10 & 10 \\
Revenue & 50 & 5 & 5 & 5 & 5 & 5 \\
\bottomrule
\end{tabular}
\end{center}

Considering selling the watch $r_1$, Alex receives
$\{r_2,r_3,r_4,r_5\}$ and Belle receives $\{r_6\}$.
Dividing the revenue proportionally (Alex~8, Belle~42)
makes both agents’ welfares equal at~52.
\end{example}
This illustrates how replacing splitting with selling can restore
fairness without compromising indivisibility, leading to our
\emph{Adjusted Winner without Splitting (AWNS)} framework.

\subsection{Problem Formulations}
We introduce the family of $\AWNS$ problems, each of which seeks to prevent the splitting of resources between agents when applying the Adjusted Winner (AW) algorithm. 

\begin{definition}[$\AWNSD$]
Given an instance of DSIRS $\mathcal{I}$, find $S_0 \subseteq R$ such that the AW-derived plan $\mathcal{S} = P_{\mathrm{AW}}(S_0, \mathcal{I})$ minimizes
\[
d = \left| W_1(\mathcal{S}, \mathcal{I}) - W_2(\mathcal{S}, \mathcal{I}) \right|.
\]
\end{definition}

\begin{definition}[$\AWNSDC$]
Given an instance of DSIRS $\mathcal{I}$ and a fixed $d \geq 0$, find $S_0 \subseteq R$ such that the AW-derived plan $\mathcal{S} = P_{\mathrm{AW}}(S_0, \mathcal{I})$ satisfies $\left| W_1(\mathcal{S}, \mathcal{I}) - W_2(\mathcal{S}, \mathcal{I}) \right| \leq d$ and minimizes
\[
c(S_0) = \sum_{r \in S_0} c(r).
\]
\end{definition}

\begin{definition}[$\AWNSR$]
Given an instance of DSIRS $\mathcal{I}$, find $S_0 \subseteq R$ such that the AW-derived plan $\mathcal{S} = P_{\mathrm{AW}}(S_0, \mathcal{I})$ minimizes
\[
\rho = \frac{\max\{ W_1(\mathcal{S}, \mathcal{I}), W_2(\mathcal{S}, \mathcal{I}) \}}{\min\{ W_1(\mathcal{S}, \mathcal{I}), W_2(\mathcal{S}, \mathcal{I}) \}}.
\]
\end{definition}

\begin{definition}[$\AWNSRC$]
Given an instance of DSIRS $\mathcal{I}$ and a fixed $\rho \geq 1$, find $S_0 \subseteq R$ such that the AW-derived plan $\mathcal{S} = P_{\mathrm{AW}}(S_0, \mathcal{I})$ satisfies
$
\frac{\max\{ W_1(\mathcal{S}, \mathcal{I}), W_2(\mathcal{S}, \mathcal{I}) \}}{\min\{ W_1(\mathcal{S}, \mathcal{I}), W_2(\mathcal{S}, \mathcal{I}) \}} \leq \rho
$
and minimizes
\[
c(S_0) = \sum_{r \in S_0} c(r).
\]
\end{definition}

The next two sections address the algorithmic and complexity aspects of the $\AWNS$ problem family. Section 4 establishes the computational complexity and inapproximability of these problems, while Section 5 presents an FPTAS for $\AWNSR$.

\section{Computational Complexity and Inapproximability}
In this section, we show that all above problems except for $\AWNSR$
are weakly NP-hard and inapproximable.
To make these hardness results precise, we first recall the formal notion of multiplicative approximation~\cite{vazirani2001approximation}.

\begin{definition}[Multiplicative Approximation]
Let \(Q\) 
be a minimization problem whose feasible solutions yield objective values \(f(\cdot)\).  
For any instance, let
\[
\text{OPT} = \min_{\text{sol} \in \mathfrak{S}_\text{ALL}} f(\text{sol}), \quad
\text{ALG} = f(\text{sol}_{\text{alg}}),
\]
where \(\text{sol}_{\text{alg}}\) is the solution returned by a polynomial-time algorithm,  $\mathfrak{S}_\text{ALL}$ is the set for all feasible solutions. Then the algorithm is called a \emph{multiplicative \(\beta\)-approximation} if
\[
\mathrm{ALG} \le \beta \cdot \mathrm{OPT},
\]
for all instances, where \(\beta \ge 1\) is the approximation ratio. If no constant $\beta$ satisfies this, we say the inapproximability holds for $Q$.

\end{definition}

This section centers on Theorem~\ref{NO-APPROX2}, which establishes the inapproximability of $\AWNSDC$.
Before the proof begins, we introduce the following two problems: 

\begin{itemize}
    \item \(\frac{1}{2}\)-Balanced Subset Sum: Given a set of integers \(H = \{a_1, a_2,  \ldots , a_n\}\) and a target sum \(B\), determine if there exists a subset \(A\) of integers from \(H\) that sums up to \(B\) and \(|A|=\lfloor \frac{1}{2}|H|\rfloor\).
    \item \(\frac{1}{k}\)-Balanced Partition: Given a set of integers \(H^* = \{a_1, a_2, \ldots , a_n\}\), determine whether there exists a subset \(A\) of integers from \(H^*\) that sums up to \(\frac{1}{2}\sum_{s\in H} s\) and \(|A|=\lfloor \frac{1}{k}|H^*|\rfloor\).
\end{itemize}
\begin{claim}
   \(\frac{1}{2}\)-Balanced Partition~\cite{garey1979computers} is weakly NP-hard, so \(\frac{1}{2}\)-Balanced Subset Sum is also weakly NP-hard.
\end{claim}
\begin{theorem}\label{1/2-BP-NP}
    For any \(k\) such that \(2 \leq k < n\), \(\frac{1}{k}\)-Balanced Partition is weakly NP-hard.
\end{theorem}
\begin{proof}
    We prove this theorem by reducing \(\frac{1}{2}\)-Balanced Subset Sum to \(\frac{1}{k}\)-Balanced Partition for any \(k\) known to be NP-hard~\cite{garey1979computers}. The input of the \(\frac{1}{2}\)-Balanced Subset Sum is \(H=\{a_1,a_2,\ldots,a_n\}\) and \(B\). We assume that \(n\) is even and \(B\leq \frac{1}{2}\sum_{s\in H} s\). For a \(\frac{1}{k}\)-Balanced Partition, we create \(H^*\) by adding normal integer \(a_i^*=a_i+2X\) (\(X\) is a large enough number) for each \(a_i\in H\). We set $k'=k-1$. In addition, we add an integer \(c=(\frac{(k'-1)n}{2}+k')X-B\) and a set \(D\) of \(\frac{(k'-1)n}{2}+k'-1\) dummy integers where each dummy integer is exactly \(X\), and a special number \(e=X-(\sum_{s\in H}s -B)\). Here, \(\sum_{s\in H^*} s=(k'n+2k'+n)X\) and $|H^*|=(k'+1)(\frac{n+1}{2})$.
 The reduction is polynomial bounded.  We now show that \(H\) is a YES-instance for \(\frac{1}{2}\)-Balanced Subset Sum if \(H^*\) is a YES-instance for \(\frac{1}{k}\)-Balanced Subset Sum:
    
    \((\Longrightarrow)\) Since \(H\) is a YES-instance, there is a subset \(A\) with \(\sum_{s\in A} s=B\) and \(|A|=\lfloor \frac{1}{2}|H|\rfloor\). Thus, we can get a set \(A^*\) which shares the same indices with \(A\) and sums up to \(nX+B\). Thus, we get a feasible solution: \(O^*=A^*\cup \{c\}\), which sums up to \((\frac{k'n+n}{2}+k)X\).  In addition, \(|O^*|=\frac{n}{2}+1\) and \(|H^*\backslash O^*|=k'(\frac{n}{2}+1)\). 
    
    \((\Longleftarrow)\) Since \(H^*\) is a YES-instance, the solution \(O^*\) must include exactly \(\frac{n}{2}+1\) integers. \(O^*\) must include \(c\). Otherwise, \(O^*\) can only sum up to at most \((n+3)X\). Besides $c$, we need to find \(\frac{n}{2}\) integers summing up to \(nX+B\). These chosen integers must be normal integers, since each integer has to provide at least \(2X\) while each of the dummy integers and the special one can provide at most \(X\). This induces a solution for \(H\).
     \end{proof}

 \begin{remark}
      When  $c(r)=1$ and $p(r)=0$ for all $r\in R$,  choosing resources to sell is exactly choosing them to "delete" and at most $B$ resources can be deleted.
 \end{remark}

\begin{theorem}\label{NO-APPROX2}
Unless $\mathrm{P} = \mathrm{NP}$, the inapproximability holds for $\AWNSDC$.
\end{theorem}

{\begin{proof}
To show inapproximability, we reduce
the $\frac{1}{k}$-Balanced Partition problem to $\AWNSDC$. Given an instance of $\frac{1}{k}$-Balanced Partition $H^*=\{a_1,a_2,\ldots,a_n\}$, we construct an instance $\mathcal{A}$ of $\AWNSDC$ as follows: (a) welfare equality is required, i.e., $d = 0$ or $\rho = 1$; (b) selling yields no return, i.e., $p(r) = 0$ for all $r \in R$; and (c) each resource has unit cost, $c(r) = 1$ for all $r \in R$, and (d) selling all the resources is not allowed, $B = n + 1$. In the following, the
first and second row refer to the utilities of resources of Agent~1 and
Agent~2 respectively.
    \[
   \begin{bmatrix}
        \begin{smallmatrix}
             a_1+m &\ldots& a_o+m&\ldots &a_l+m&\ldots & a_n+m & m^* & m^*\\
            0  &\ldots &0 &\ldots& 0 &\ldots&  0  & s_\text{half}+lm+m^* & s_\text{half}+om+m^* 
        \end{smallmatrix}
    \end{bmatrix}
    \]
    with \(l=n-\lfloor\frac{n}{k}\rfloor\), \(o=\lfloor\frac{n}{k}\rfloor\), \(m\) is a large enough number, \(m^*=mn\), \(s=\sum^{n}_{i=1}{a_i}\), and \(s_\text{half}=\frac{1}{2}s\). Here, w.l.o.g., we assume $\frac{n}{k}$ and $s_\text{half}$ are integers, s.t. \(l=(k-1)o\) and \(s\) is even.  The reduction can be performed in polynomial time. Now we have further observations.

    \begin{itemize}
        \item Exactly one of the last two resources should be kept for Agent~2. Otherwise, Agent~2 gets either 0 or a too large number. The chosen one is kept for Agent~2, and the unchosen one cannot be deleted. 
        \item Instead, the unchosen one should be transferred to Agent~1. This can ensure that Agent~1 gets at least \(m^*\).
        \item If the last but one is transferred, then exactly \(l\) resources have to be deleted; if the last one is transferred, then exactly \(o\) resources have to be deleted. The latter is the optimal solution.
    \end{itemize}
    
     We show that \(H^*\) is a YES-instance if and only if \(\mathcal{A}\) is a YES-instance.
    
    \((\Longrightarrow)\) Since \(H^*\) is a YES-instance, there is a subset \(A\) with \\ \(\sum_{e\in A}e =\frac{1}{2}\sum_{e\in H^*}e= s_\text{half}\) and \(|A|=\lfloor \frac{1}{k}|H^*|\rfloor=o\). Then, we get two solutions for \(\mathcal{A}\): (1) Transferring the last resource and deleting exactly \(o\) other resources corresponding to the resources in \(A\). In this case, either of the two agents gets \(s_\text{half}+lm+m^*\). (2) Transferring the last but one resource and deleting exactly \(l\) other resources corresponding to the resources in \(H^*\backslash A\). In this case, either of the two agents gets \(s_\text{half}+om+m^*\).
    
    \((\Longleftarrow)\) Since \(\mathcal{A}\) is a YES-instance, there are at least two solutions: (a) Transferring the last resource and deleting exactly \(o\) other resources. (b) Transferring the last but one resource and deleting exactly \(l\) other resources. The solution (a) induces an optimal solution for \(H^*\).\par
Assume, for contradiction, that there exists a polynomial-time approximation algorithm for $\AWNSDC$ with approximation ratio $\beta$. For any instance of $\frac{1}{k}$-Balanced Partition with $k-1 \leq \beta$, the constructed $\AWNSDC$ instance has $\frac{l}{o} = k-1 \leq \beta$, where $o$ is the number of resources deleted in the optimal solution and $l$ in the only alternative.
 Thus, the algorithm would return a solution with $l$ resources to delete, thereby indirectly distinguishing between YES- and NO-instances of the $\frac{1}{k}$-Partition problem in polynomial time, contradicting its NP-hardness unless $\mathrm{P} = \mathrm{NP}$.

In summary, to achieve $d = 0$ in the constructed instance of $\AWNSDC$, there are at most two feasible solutions: the one deletes $o$ resources (the optimal) and the other deletes $l$ resources; if $\mathcal{A}$ is a YES-instance, such two solutions exist, and otherwise no solution exists. The construction guarantees that only these two solutions are possible, with $\frac{l}{o}=k-1$ (which is not constant). Therefore, no polynomial-time algorithm can achieve a constant approximation ratio $\beta$ for $\AWNSDC$ unless $\mathrm{P} = \mathrm{NP}$.
     \end{proof}}

Combining Theorem \ref{NO-APPROX2} and Corollary \ref{easydrho} yields the following corollary.
\begin{corollary}
Unless $\mathrm{P} = \mathrm{NP}$, the inapproximability holds for $\AWNSDC$, $\AWNSRC$, even when $d=0,\rho=1$.
\end{corollary}

In the proof of Theorem~\ref{NO-APPROX2}, we observe that the instance of the $\frac{1}{k}$-Balanced Partition problem is a YES-instance if and only if there exists an allocation for which $d = 0$ and $\rho = 1$---that is, exactly $o$ resources are sold. Evidently, the same reduction applies to show that $\AWNSD$, and {$\AWNSDC$} are both weakly NP-hard.

\begin{corollary}
    \label{NP-20}
  $\AWNSD$, $\AWNSR$ are both weakly NP-hard, and even when $d=0,\rho=1$ is given, $\AWNSDC$, $\AWNSRC$, are also weakly NP-hard.
 
\end{corollary}

\begin{theorem}\label{InAppro2}
Unless $\mathrm{P} = \mathrm{NP}$, the inapproximability holds for  $\AWNSD$.
\end{theorem}

{
{
\begin{proof}
    Towards a contradiction, assume such an approximation algorithm $A$ with an approximation ratio \(\beta\) exists. For the whole set of instances of AWNS-\(d\), we can divide them into two sets: \(\mathcal{I}_0\) and \(\mathcal{I}_+\). \(\mathcal{I}_0\) stands for the instances where $d=0$ is feasible and \(\mathcal{I}_+\) stands for the instances where $d=0$ is not feasible.
    
    For any instance \(I\in \mathcal{I}_0\), the output of the approximation algorithm must be 0. Otherwise, the approximation ratio is infinity. For any instance \(I\in \mathcal{I}_+\), the output of the approximation algorithm \(A\) must not be 0. Otherwise, the algorithm is not correct. Thus, the algorithm \(A\) is a polynomial algorithm that would decide whether $d=0$ is feasible for \(I\). However, we already know that AWNS-\(d\) with $d=0$ and \(k<n\) is weakly NP-hard from Corollary \ref{NP-20}. Unless $\mathrm{P}=\mathrm{NP}$, there is a contradiction.
     \end{proof}
}
}

\section{An FPTAS for  $\AWNSR$}
We continue to show an FPTAS for  $\AWNSR$.  A fully polynomial-time approximation scheme (FPTAS) is widely considered a desirable form of approximation algorithm, as it offers a flexible trade-off between computational efficiency and solution accuracy for many optimization problems that are typically regarded as hard.\cite{vazirani2001approximation}
\begin{definition}[FPTAS]
A \emph{fully polynomial-time approximation scheme} (FPTAS) for an optimization problem is an algorithm that, given any instance and any $\varepsilon > 0$, produces a solution with value at most $(1+\varepsilon)$ times the optimal value (for minimization problems), or at least $(1-\varepsilon)$ times the optimal value (for maximization problems), and runs in time polynomial in both the input size and $1/\varepsilon$.
\end{definition}

Before proceeding with the proof, we offer the following remarks to help clarify the subsequent argument.

\begin{remark}
    Here, we rename the set $R$ so that the elements $r_i$ are ordered in descending order according to the ratio $\frac{u_{i1}}{u_{i2}}$. In addition, we introduce two indices, denoted \( i_{\blacktriangleleft} \) and \( i_{\blacktriangleright} \) (with \( i_{\blacktriangleleft} \leq i_{\blacktriangleright} \)), which indicate the positions where transfers or sales must occur.
\end{remark}
\begin{remark}
    Since the AW algorithm does not specify how to choose the transfer sequence when multiple resources share the same value of $\frac{u_{i1}}{u_{i2}}$, we assume that the ratio-sorted list of resources follows a fixed tie-breaking rule: ties are resolved from left to right according to a predefined order (ties are resolved by resource index in increasing order).
\end{remark}

\begin{theorem}\label{AWNS-FPTAS-DP}
    There is an FPTAS for AWNS-${\rho}$.
\end{theorem}

{
\begin{proof}

 (A) First, we design a dynamic programming  for  AWNS-${\rho}$.  We define a table $T[\mathfrak{i}, \mathfrak{o}, \mathfrak{u}_1, \mathfrak{u}_2, \mathfrak{p}]$, where each entry corresponds to a tuple representing different aspects of the problem. The entries of the tuple are as follows:

\begin{itemize}
    \item $\mathfrak{i}$: Consider the first $\mathfrak{i}$ resources.
    \item $\mathfrak{o}$: The number of resources to be  sold.
    \item $\mathfrak{u}_1$: The current utility for Agent~1 with respect to allocation.
    \item $\mathfrak{u}_2$: The current utility for Agent~2 with respect to allocation.

    \item $\mathfrak{p}$: The current revenue obtained from selling the resources.
\end{itemize}

 With this table, we can systematically explore the different combinations and transitions of resources being allocated and sold, taking into consideration the utilities and costs for both agents, as well as the revenue obtained from sales. $T[\mathfrak{i}, \mathfrak{o}, \mathfrak{u}_1, \mathfrak{u}_2,\mathfrak{p}] = \,\perp\text{(undefined)}$ means that the combination is impossible. If it equals a tuple consisting of three sets $S_0,S_1,S_2$, and a number for the cost $\mathfrak{c}$, i.e. $[S_0, S_1, S_2,\mathfrak{c}]$, then the combination is possible, and these three sets are the solution and $\mathfrak{c}$ is the cost:$c(S_0)$. Notice that the entry  $T[\mathfrak{i},\mathfrak{o},\mathfrak{u}_1,\mathfrak{u}_2,\mathfrak{p}]$  only aims to record the solution with the lowest cost $\mathfrak{c}$.  
Here is the base of the algorithm:

\begin{itemize}
    \item $T[\mathfrak{i},\mathfrak{o},\mathfrak{u}_1,\mathfrak{u}_2,\mathfrak{p}]=\,\perp$ if any of $\mathfrak{i},\mathfrak{o},\mathfrak{u}_1,\mathfrak{u}_2,\mathfrak{p}$ is negative or $\mathfrak{i}<\mathfrak{o}$,
    \item $T[0,0,0,0,0]=[\emptyset,\emptyset,\emptyset,0]$,
    \item $T[0,\mathfrak{o},\mathfrak{u}_1,\mathfrak{u}_2,\mathfrak{p}]=\,\perp$ if any of $\mathfrak{o},\mathfrak{u}_1,\mathfrak{u}_2,\mathfrak{p}\neq 0$.
\end{itemize}

 \begin{remark}
Note that we divide $R$ into $\Rof1$, $\Rof2$, and $\Rof0$: $r\in \Rof0$ if $u_1(r)=u_2(r)$, and $r\in \Rof{j}$ if $u_{j}(r)>u_{(3-j)}(r)$ for $j\in\{1,2\}$. We further define $A_1=\Rof1,\,A_2=\Rof2\cup\Rof0$.
 \end{remark} 
 
\begin{remark}
For each instance $\mathcal{I}$ of AWNS-${\rho}$, there must be a solution $\langle S_0,S_1,S_2\rangle$ minimizing the ratio $\rho^\heartsuit=\frac{W_1}{W_2} (W_1\geq W_2)$ and $c(S_0)\leq B$. Here, we stipulate that for this \( \mathcal{I} \), \( W_1 \) is always greater than \( W_2 \). If this is not possible, we interchange the names of Agent~1 and Agent~2. If no renaming occurs, we define \( A_1 = \Rof1 \) and \( A_2 = \Rof2 \cup \Rof0 \); if renaming occurs, we define \( A_1 = \Rof1 \cup \Rof0 \) and \( A_2 = \Rof2 \).

\end{remark}

The induction rules of transferring from Agent~1 to Agent~2  and transferring from Agent~2 to Agent~1 are given in the following. Additionally, we define two indices, $i^\blacktriangleleft$ and $i^\blacktriangleright$, such that resource transfer is restricted to the interval between them, within which every resource is either transferred or sold - none are retained. During computation, DP algorithm tries out all the possible choices of positions of  $i^\blacktriangleleft$ and $i^\blacktriangleright$ as well as the direction of transfer - whether it is from $A_1$ to $A_2$ (1) or from $A_2$ to $A_1$ (2).

\begin{align}
T[\mathfrak{i},\mathfrak{o},\mathfrak{u}_1,\mathfrak{u}_2,\mathfrak{p}]
=
\begin{cases}
\big[S_0\cup\{r_i\},\,S_1,\,S_2,\,\mathfrak{c}+c_{i0}\big],
& \begin{aligned}[t]
  \text{if }&
  T[\mathfrak{i}-1,\mathfrak{o}-1,\mathfrak{u}_1,\mathfrak{u}_2,\mathfrak{p}-p_i]
  = [S_0,S_1,S_2,\mathfrak{c}]\\
  &\land\ \mathfrak{c}+c_{i0}\le B,
\end{aligned}
\\[4pt]
\big[S_0,\,S_1\cup\{r_i\},\,S_2,\,\mathfrak{c}+c_{i1}\big],
& \begin{aligned}[t]
  \text{if }&
  T[\mathfrak{i}-1,\mathfrak{o},\mathfrak{u}_1-u_{i1},\mathfrak{u}_2,\mathfrak{p}]
  = [S_0,S_1,S_2,\mathfrak{c}]\\
  &\land\ \mathfrak{c}+c_{i1}\le B\\
  &\land\ (r_i\in A_1 \text{ and } i\le i^\blacktriangleleft),
\end{aligned}
\\[4pt]
\big[S_0,\,S_1,\,S_2\cup\{r_i\},\,\mathfrak{c}+c_{i2}\big],
& \begin{aligned}[t]
  \text{if }&
  T[\mathfrak{i}-1,\mathfrak{o},\mathfrak{u}_1,\mathfrak{u}_2-u_{i2},\mathfrak{p}]
  = [S_0,S_1,S_2,\mathfrak{c}]\\
  &\land\ \mathfrak{c}+c_{i2}\le B\\
  &\land\ (r_i\in A_2 \text{ or } i\ge i^\blacktriangleleft).
\end{aligned}
\end{cases}
\end{align}

\begin{align}
T[\mathfrak{i},\mathfrak{o},\mathfrak{u}_1,\mathfrak{u}_2,\mathfrak{p}]
=
\begin{cases}
\big[S_0\cup\{r_i\},\,S_1,\,S_2,\,\mathfrak{c}+c_{i0}\big],
& \begin{aligned}[t]
  \text{if }&
  T[\mathfrak{i}-1,\mathfrak{o}-1,\mathfrak{u}_1,\mathfrak{u}_2,\mathfrak{p}-p_i]
  = [S_0,S_1,S_2,\mathfrak{c}]\\
  &\land\ \mathfrak{c}+c_{i0}\le B,
\end{aligned}
\\[4pt]
\big[S_0,\,S_1\cup\{r_i\},\,S_2,\,\mathfrak{c}+c_{i1}\big],
& \begin{aligned}[t]
  \text{if }&
  T[\mathfrak{i}-1,\mathfrak{o},\mathfrak{u}_1-u_{i1},\mathfrak{u}_2,\mathfrak{p}]
  = [S_0,S_1,S_2,\mathfrak{c}]\\
  &\land\ \mathfrak{c}+c_{i1}\le B\\
  &\land\ (r_i\in A_1 \text{ or } i\le i^\blacktriangleright),
\end{aligned}
\\[4pt]
\big[S_0,\,S_1,\,S_2\cup\{r_i\},\,\mathfrak{c}+c_{i2}\big],
& \begin{aligned}[t]
  \text{if }&
  T[\mathfrak{i}-1,\mathfrak{o},\mathfrak{u}_1,\mathfrak{u}_2-u_{i2},\mathfrak{p}]
  = [S_0,S_1,S_2,\mathfrak{c}]\\
  &\land\ \mathfrak{c}+c_{i2}\le B\\
  &\land\ (r_i\in A_2 \text{ and } i\ge i^\blacktriangleright).
\end{aligned}
\end{cases}
\end{align}

In addition, it follows the following restriction, the Remark \ref{red_restrict}: 
\begin{remark}
\label{red_restrict}
    Assume $T[\mathfrak{i},\mathfrak{o},\mathfrak{u}_1,\mathfrak{u}_2,\mathfrak{p}]= [S_0,S_1,S_2,\mathfrak{c}]$, it will be updated to $[S_0',S_1',S_2',\mathfrak{c'}]$ only if $\mathfrak{c'}<\mathfrak{c}$. This restriction ensures that the entry  $T[\mathfrak{i},\mathfrak{o},\mathfrak{u}_1,\mathfrak{u}_2,\mathfrak{p}]$  only records the solution with the lowest cost. If $T[\mathfrak{i},\mathfrak{o},\mathfrak{u}_1,\mathfrak{u}_2,\mathfrak{p}]= \perp\text{(undefined)}$, we treat its cost as a very large number and an update takes place only when the candidate cost is strictly less than the current stored cost.  This restriction also applies for the inductions rule in the following proofs.
\end{remark}

\begin{remark}
Here, $\langle S_0,S_1,S_2\rangle$ is the solution minimizing the ratio $\rho^\heartsuit=\frac
{W_1}{W_2}$ of the instance $\mathcal{I}$. There must be a resource with index $j_1$ for $S_1$ such that $u_{j_1 1}=\max\{u_{i1}\in S_1\}=u_{\max,1}$ and a resource with index $j_2$ for $S_2$ such that $u_{j_2 2}=\max\{u_{i2}\in S_2\}=u_{\max,2}$, and a resource with index $j_0$ for $S_0$ such that $p_{j_0}=\max\{p_{i}\in S_0\}=p_{\max}$. Finally, we have $m_{\max}=\lceil\frac{1}{2}\max\{p_{\max}, u_{\max,1},u_{\max,2}\}\rceil$.
\end{remark}

\begin{remark}
    We assume that the ratio-sorted list of resources is equipped with a fixed tie-breaking rule: when multiple resources have the same ratio, they are ordered from left to right according to a predefined order. The AW algorithm does not strictly define how to choose the transfer sequence when  different resources have the same  $\frac{u_{i1}}{u_{i2}}$ value.
\end{remark}

(B) We define the ceiling of half of the maximum value that occurs in the optimal solution: $m_{\max}$. Here, $\mathcal{S}=\langle S_0,S_1,S_2\rangle$ is the solution minimizing the ratio $\rho^\heartsuit=\frac
{W_1}{W_2}$ of the instance $\mathcal{I}$. In this proof, we assume $W_1\geq W_2$. (If this is not possible, swap the roles of Agent~1 and Agent~2 and consider both Phase~1 variants in which all utilities are initially assigned to Agent~1 or to Agent~2 (yielding four cases in total).)  There must be a resource with index $j_1$ for $S_1$ such that $u_{j_1 1}=\max\{u_{i1}\in S_1\}=u_{\max,1}$ and a resource with index $j_2$ for $S_2$ such that $u_{j_2 2}=\max\{u_{i2}\in S_2\}=u_{\max,2}$, and a resource with index $j_0$ for $S_0$ such that $p_{j_0}=\max\{p_{i}\in S_0\}=p_{\max}$. Finally, we have $m_{\max}=\lceil\frac{1}{2}\max\{p_{\max}, u_{\max,1},u_{\max,2}\}\rceil$.

We apply scaling to the parameters and carefully choose whether to round them up or down, ensuring that the size of the dynamic programming table remains within polynomial bounds. Here, the selection of upper and lower bounds for scaling
is rigorous to ensure the validity of the subsequent inequalities. The numerical magnitude of the scaled values must be just right; if too large, they may not be polynomially bounded, and if too small, the
inequalities in  (C) may become difficult or impossible to prove. For the given instance $\mathcal{I}$, we rescale it in the following way to get a new instance  $\mathcal{\hat{I}}$:

\begin{enumerate}
    \item For each $i\in\{1,\ldots,n\}$: if $u_{i1}\leq u_{\max,1}$, then $\hat{u_{i1}}\leftarrow \lceil \frac{u_{i1}\cdot n}{\epsilon' \cdot m_{\max}} \rceil$; otherwise, then $\hat{u_{i1}}\leftarrow \infty$; 
       \item  For each $i\in\{1,\ldots,n\}$: if $u_{i2}\leq u_{\max,2}$, then $\hat{u_{i2}}\leftarrow \lfloor\frac{u_{i2}\cdot n}{\epsilon' \cdot m_{\max}} \rfloor$; otherwise, then $\hat{u_{i2}}\leftarrow -\infty$;
       \item  For each $i\in\{1,\ldots,n\}$: if $p_{i}\leq p_{\max}$, then $\hat{p_i}\leftarrow \lfloor\frac{p_{i}\cdot n}{\epsilon' \cdot m_{\max}} \rfloor$; otherwise, then $\hat{p_{i}}\leftarrow -\infty$.
        \item  For each $i\in\{1,\ldots,n\}$: if $p_{i}\leq p_{\max}$, then $\hat{p_{i+}}\leftarrow \lceil\frac{p_{i}\cdot n}{\epsilon' \cdot m_{\max}} \rceil$; otherwise, then $\hat{p_{i+}}\leftarrow \infty$.

    \end{enumerate}

(C)  Among all the solutions $\langle S_0^+,S_1^+,S_2^+\rangle$ satisfying the constraint: $$\lceil\frac{1}{2}\max\{\max\limits_{r_1\in {S_1^+}}\{u_{i1}\},\max\limits_{r_i\in {S_2^+}}\{u_{i2}\},\max\limits_{r_i\in {S_0^+}}\{p_{i}\}\}\rceil=m_{\max} \quad (\bigstar)$$ ($m_{\max}$ is still the ceiling of half of the maximum value that occurs in the optimal solution $\mathcal{S}=\langle S_0,S_1,S_2\rangle$), let $\mathcal{S'}=\langle{S_0'},{S_1'},{S_2'}\rangle$ be the optimal one for the scaled instance $\mathcal{\hat{I}}$.

Let $\rho^\heartsuit$ be the ratio of $\mathcal{S}$ wrt. $\mathcal{I}$,  $\rho^\diamondsuit$ be the ratio of $\mathcal{S}$ wrt. $\mathcal{\hat{I}}$, Let $\rho^\clubsuit$ be  the ratio of $\mathcal{S'}$ wrt. $\mathcal{I}$ and $\rho^\spadesuit$ be  the ratio of $\mathcal{S'}$ wrt. $\mathcal{\hat{I}}$.

\begin{table}[]
    \centering
    \begin{tabular}{|c|c|c|}
        \hline
              &  \makecell{ $\mathcal{I}$\\(Original Instance)} &  \makecell{$\mathcal{\hat{I}}$\\(Scaled Instance)  }\\
        \hline
         \makecell{ $\mathcal{S}$\\ (Optimal Solution of\\ Original Instance)} & $\rho^\heartsuit$  & $\rho^\diamondsuit$ \\
         \hline
         \makecell{ $\mathcal{S'}$ \\ (Optimal Solution of\\ Scaled Instance s.t $(\bigstar)$} & $\rho^\clubsuit$ & $\rho^\spadesuit$\\
        \hline

    \end{tabular}
\end{table}

 \begin{remark}
          Given the instance $\mathcal{I}$, the function $O_2:2^R\rightarrow \mathbb{N}^+$ refers to the optimal welfare for Agent~2 given $R_S\subseteq R$. The resources in $R_S$ are either allocated to Agent~2 or sold. In addition, Agent~2 will get all the revenue.
      \end{remark}

Proceeding forward, to prove that ${\rho^\clubsuit}\leq (1+\epsilon){\rho^\heartsuit}$, we will distinguish between two cases: one in which no resource $r_i\in R$ with $u_{i1}> 2O_2(R\backslash\{r_i\})$, and the other in which such a resource is present. We refer to such a resource as a \emph{$u_1$-heavy resource}. The two cases correspond to Lemma~\ref{TRatio1} and Lemma~\ref{TRatio2}, respectively.

  \begin{lemma}\label{TRatio1}

If there is no resource $r_i\in R$ with $u_{i1}> 2O_2(R\backslash\{r_i\})$ (no {$u_1$-heavy resource}), ${\rho^\clubsuit}\leq (1+\epsilon){\rho^\heartsuit}$. 
\end{lemma}

\begin{proof}
  Let ${\rho^\heartsuit}=\frac{\left(\sum\limits_{r_i\in S_1} u_{i1} +  q \left(\sum\limits_{r_i\in S_0} p_{i}\right)\right)}{\left(\sum\limits_{r_i\in S_2} u_{i2}+(1-q) \left(\sum\limits_{r_i\in S_0} p_{i}\right)\right)}$ be the optimal ratio of $\mathcal{I}$ and $S_i$ be the set of resources allocated to agent $i$ and $S_0$ is the set of resources to be sold. We set: 
  \begin{itemize}
      \item $K=\frac{\epsilon'\cdot m_{\max}}{n}$ (Scaling Parameter),
      \item $\epsilon'=\frac{\epsilon}{\epsilon+2}$ (so that  $1+\epsilon=\frac{1+\epsilon'}{1-\epsilon'}$),
      \item $U_1=\sum\limits_{i\in S_1} u_{i1}$ (the utility for Agent~1 from $S_1$), 
      \item $U_2=\sum\limits_{i\in S_2} u_{i2}$ (the utility for Agent~2 from $S_2$), 
      \item $P=\sum\limits_{i\in S_0} p_{i}$ (the revenue from $S_0$),
      \item $\hat{U_1}=\sum\limits_{i\in S_1} \hat{u_{i1}}$ (the scaled version for $U_1$),
      \item $\hat{U_2}=\sum\limits_{i\in S_2} \hat{u_{i2}}$ (the scaled version for $U_2$),
      \item $\hat{P}=\sum\limits_{i\in S_0} \hat{p_{i}}$ (the scaled version for $P$, round down),
      \item $\hat{P_{+}}=\sum\limits_{i\in S_0} \hat{p_{i+}}$ (the scaled version for $P$, round up), 
      \item $\rho^\heartsuit=\frac{U_1+qP}{U_2+(1-q)P}=\frac{W_1}{W_2}$,
      \item $\rho^\diamondsuit=\frac{\hat{U_1}+\hat{q}\hat{P}}{\hat{U_2}+(1-\hat{q})\hat{P}}$. (Here, the choice of $\hat{q}\in [0,1]$ minimizes  $\rho^\diamondsuit$. Thus, $\frac{\hat{U_1}+\hat{q}\hat{P}}{\hat{U_2}+(1-\hat{q})\hat{P}}\leq \frac{\hat{U_1}+{q}\hat{P}}{\hat{U_2}+(1-{q})\hat{P}}$.)
  \end{itemize} 
  Following the aforementioned definitions, we define $U_1',U_2',P'$ and $\hat{U_1'},\hat{U_2'},\hat{P'},\hat{P_+'}$ for the solution $\mathcal{S'}=\langle S_0',S_1',S_2'\rangle$, with the key aspect being the replacement of $S_0,S_1,S_2$ with $S_0',S_1',S_2'$ respectively. Here, we have:
  \begin{itemize}
           \item $\rho^\clubsuit=\frac{U_1'+q_1 P'}{U_2'+(1-q_1)P'}$,  (Here, the choice of ${q_1}\in [0,1]$ minimizes  $\rho^\clubsuit$.)
      \item $\rho^\spadesuit=\frac{\hat{U_1'}+\hat{q_1}\hat{P_+'}}{\hat{U_2'}+(1-\hat{q_1})\hat{P'}}$ (Here, the choice of $\hat{q_1}\in [0,1]$ minimizes  $\rho^\spadesuit$.)
  \end{itemize}
  
  We are going to prove that: \begin{enumerate}
      \item $K\cdot (\hat{U_1}+q\hat{P_{+}})\leq (1+\epsilon')(U_1+qP)$ 
      \begin{align*}
          &K\cdot (\hat{U_1}+q\hat{P_{+}}) \\=& K\cdot \left(\sum\limits_{i\in S_1} \hat{u_{i1}}+\sum\limits_{i\in S_0} q\hat{p_{i+}}\right)\\ \leq&
          K\cdot \left(\frac{1}{K}\sum\limits_{i\in S_1} {u_{i1}}+\frac{1}{K}(\sum\limits_{i\in S_0} {qp_{i}})+n\right)\\=& \sum\limits_{i\in S_1} {u_{i1}}+\sum\limits_{i\in S_0} {qp_{i}}+\epsilon'\cdot m_{\max}\\
          \leq&^{(a)} \left(\sum\limits_{i\in S_1} {u_{i1}}+\sum\limits_{i\in S_0} {qp_{i}}\right) (1+\epsilon')
        \\=& (1+\epsilon')(U_1+qP).
         \end{align*}
     \item $K(\hat{U_2}+(1-q)\hat{P})\geq(1-\epsilon')(U_2+(1-q)P)$.
          \begin{align*}
         & K\cdot (\hat{U_2}+(1-q)\hat{P})\\=& K\cdot \left(\sum\limits_{i\in S_2} \hat{u_{i1}}+\sum\limits_{i\in S_0}(1-q) \hat{p_{i}}\right)\\ \geq&
          K\cdot \left(\frac{1}{K}\sum\limits_{i\in S_2} {u_{i2}}+\frac{1}{K}(\sum\limits_{i\in S_0} {(1-q)p_{i}})-n\right)\\=&\sum\limits_{i\in S_2} {u_{i2}}+\sum\limits_{i\in S_0} {(1-q)p_{i}}-\epsilon'\cdot m_{\max}\\
          \geq&^{(b)} \left(\sum\limits_{i\in S_2} {u_{i2}}+\sum\limits_{i\in S_0} {(1-q)p_{i}}\right) (1-\epsilon')\\
          =& (1-\epsilon')(U_2+(1-q)P).
          \end{align*}
  \end{enumerate}

 Here, we set $\epsilon=\frac{2\epsilon'}{1-\epsilon'}$, so that $\epsilon'=\frac{\epsilon}{\epsilon+2}$. Thus,
\begin{align*}
  \rho^\diamondsuit &= \frac{\hat{U_1} + \hat{q}\hat{P}}{\hat{U_2} + (1-\hat{q})\hat{P}} \\
  &\leq \frac{\hat{U_1} + q\hat{P}}{\hat{U_2} + (1-q)\hat{P}} \\
  &\leq  \left(\frac{1+\epsilon'}{1-\epsilon'}\right)\frac{U_1 + qP}{U_2 + (1-q)P}  \\
  &= \left(\frac{1+\epsilon'}{1-\epsilon'}\right) \rho^\heartsuit\\
  &= (1+\epsilon)\rho^\heartsuit  
\end{align*}

  Now, we need to make case distinction to prove the inequalities $(a)$ and $(b)$ hold: (1) $0<q\leq 1$ or (2) $q=0$ and $m_{\max}=\frac{1}{2}u_{\max,2}$ or  $m_{\max}=\frac{1}{2}p_{\max}$  or (3) $q=0$ and $m_{\max}=\frac{1}{2}u_{\max,1}$.
  \begin{itemize}
      \item Case 1: Since $0<q\leq 1$, $\rho^\heartsuit=1$, we have $U_1+qP=U_2+(1-q)P>m_{\max}$, which means here two agents have the same welfare. In this case, the revenue is split into two parts, which means that the gap between the profits of two agents is smaller than the revenue of selling resources. 
      \item Case 2: Since $q=0$  and $U_1+qP\geq U_2+(1-q)P$, we have $U_1>U_2+P$,  which means here Agent~1 has a higher welfare than Agent~2 even when the gap is smallest. Since $m_{\max}=\frac{1}{2}u_{\max,2}$ or  $m_{\max}=\frac{1}{2}p_{\max}$ or $m_{\max}=\frac{1}{2}x$, $U_1>U_2+P>m_{\max}$. This is because, in this case, $u_{\max,2},p_{\max},x$ are parts of composition of Agent~2's welfare.
 
      \end{itemize}
     
      \begin{itemize}
      \item Case 3: $q=0$ and $m_{\max}=\frac{1}{2}u_{\max,1}$.
      \item Case 3.1: The inequality (a) (b) holds immediately when there is no resource $r_i$ with $u_{i1}> 2O_2(R\backslash\{r_i\})$ (no $u_1$-heavy resource). 
      
  \end{itemize}

Again, we show that: 
\begin{align*}
\rho^\diamondsuit\leq {\rho}^\heartsuit(1+\epsilon)
\end{align*}

According to the optimality from the definition of $\rho^\spadesuit$, we have:
\begin{align*}
  \rho^\spadesuit\leq \rho^\diamondsuit  
\end{align*}
Clearly, due to the proper selection of rounding up and down, we have:
\[
  \rho^\clubsuit=\frac{U_1'+q_1P'}{U_2'+(1-q_1)P'}\leq \frac{U_1'+\hat{q_1}P'}{U_2'+(1-\hat{q_1})P'} \leq 
\]
\[
\frac{\hat{U_1'}+\hat{q_1}\hat{P_+'}}{\hat{U_2'}+(1-\hat{q_1})\hat{P'}} =\rho^\spadesuit  
\]

Thus, we complete that:
\begin{align*}
\rho^\clubsuit\leq\rho^\spadesuit\leq\rho^\diamondsuit\leq {\rho}^\heartsuit(1+\epsilon).
\end{align*}
\end{proof}

(D) It remains to consider the case where there exists a {$u_1$-heavy resource}. In this case, we still have ${\rho^\clubsuit} \leq (1+\epsilon){\rho^\heartsuit}$.
Please notice that if such resource $r_i$ with $u_{i1}> 2O_2(R\backslash\{r_i\})$ exists, there can be at most one. 
\begin{lemma}
     There exists at most one {$u_1$-heavy resource} in every instance.
\end{lemma}
\begin{proof}
    Assume that there are a pair of resources $r_x,r_y$ satisfying $u_{x1}> 2O_2(R\backslash\{r_x\})$ and $u_{y1}> 2O_2(R\backslash\{r_y\})$. Thus, $u_{x1}+u_{y1}>2(u_{x2}+u_{y2}+2\sum\limits_{i\in R\backslash\{x,y\}}u_i)>2(u_2(R))=2M$. We know $u_{x1}+u_{y1}<u_1(R)=M$. Contradiction.
     \end{proof}

\begin{lemma}
\label{TRatio2}
     Even if there is a {$u_1$-heavy resource} $r_i$, ${\rho^\clubsuit}\leq (1+\epsilon){\rho^\heartsuit}$ still holds. 
\end{lemma}
\begin{proof}
    We just need to make a simple case distinction:
\begin{itemize}
    \item Case 3.2: $r_i$ is allocated to Agent~1. Then, for all the remaining resources, allocate them to Agent~2 or sell them and set $q=1$.
    \item Case 3.3:  $r_i$ is sold, then create a new instance by setting $R=R\backslash\{r_i\}\cup \{r_i^\star\}$, $B=B-c_i$, where $u_1(r_i^\star)=u_2(r_i^\star)=c(r_i^\star)=0, p(r_i^\star)=p(r_i).$ It is straightforward to observe that selling $r_i^\star$ yields a strictly preferable (Pareto-dominant) outcome compared to allocating it---unless $p(r_i^\star) = 0$, in which case selling and allocating are equivalent. Our dynamic programming algorithm explicitly accounts for such situations.

\end{itemize}

The above steps can be repeated until the new instance no longer has such a resource or the $B$ is no longer positive. Theorem \ref{ageng2_knap}  tells us the problem of computing $O_2$ also has an FPTAS. Therefore, Cases 3.2 and 3.3 do not contradict the fact that this dynamic programming is an FPTAS.
     \end{proof}

Next, we show that the special case mentioned in Lemma \ref{TRatio2} (Case 3.2) and the knapsack problem can be mutually reduced via a polynomial-time many-one reduction. Consequently, this special case also admits an FPTAS, inherently included in the dynamic programming approach. Finally, we will analyze the time complexity to confirm that this method is indeed an FPTAS.
\begin{definition}[0/1-Knapsack]
   Given a set of resources \\$R=\{s_1,\ldots,s_n\}$, the 0/1-Knapsack Problem is defined as maximizing the total value $\sum\limits_{i=1}^{n} v_i x_i$ subject to the constraint $\sum\limits_{i=1}^{n} w_i x_i \leq W$, where $x_i \in \{0, 1\}$ for $i = 1, 2, \ldots, n$. Here, $v_i$ and $w_i$ are the value and the weight of the resource $s_i$ respectively, and $W$ is the knapsack's capacity.
\end{definition}

\begin{lemma}\label{EQUIP-KNAP}
 Given an instance of DSIRS $\mathcal{I}$ and a set of resources $S\subseteq R$ and $q=0$, we call the problem, to find an optimal solution for deciding selling or allocating to Agent~2 for each resource in $S$ to maximize the welfare for Agent~2, $O_2$ problem. $O_2$ problem has an FPTAS, because this problem is equivalent to the 0/1-Knapsack problem.
 \label{ageng2_knap}
\end{lemma}

{
\begin{proof}
    We prove this equivalence with two steps.
    \begin{itemize}
        \item (0/1-Knapsack $\Rightarrow$ $O_2$) This direction is trivial. For each resource $r_i$ with $v_i$ and $w_i$ of a 0/1-Knapsack $\mathcal{K}$, we create a resource $r_i$ with the revenue $p_i=v_i$ and the cost $c_i=w_i$ respectively in a $O_2$ instance $\mathcal{O}$. Additionally, we set $u_{i2}=0$ to each resource $r_i \in\mathcal{O}$ and  $B=W$.
        \item ($O_2$ $\Rightarrow$ 0/1-Knapsack) Given an $O_2$ instance $\mathcal{O}$, we ignore all resources with $u_{i2}\geq p_i$ in the first step, because it is always better to allocate them to Agent~2. In the second step, for each resource $r_{i}$ with $u_{i2}< p_i$ and the cost $c_i$ in $\mathcal{O}$, we create a resource $s_i$ with the value $v_i=p_i$ and the weight $c_i=w_i$ respectively a 0/1-Knapsack instance in $\mathcal{K}$ and set $W=B$.
    \end{itemize}
These reductions are clearly polynomially bounded. Since there is a well-known FPTAS of the 0/1-Knapsack with time complexity $\mathcal{O}(n^2\lfloor\frac{n}{\epsilon}\rfloor)$ \cite{vazirani2001approximation} and these two problems are equivalent, $O_2$ problem has an FPTAS.
     \end{proof}
}

(E) At the end, we analyze the time complexity.
\begin{lemma}\label{AWNS-FPTAS}
    The time complexity of this dynamic programming algorithm with scaling is polynomial in the input size and in $\frac{1}{\epsilon}$.
\end{lemma}
\begin{proof}

For each instance of AWNS-$\rho$, we enumerate all possible combinations of $j_1$, $j_2$, $j_0$, and $m_{\max}$, resulting in $\mathcal{O}(n^3)$ scaled instances for $\mathcal{I}$. For each scaled instance, the dynamic programming (DP) table is indexed by $[{\mathfrak{i}}, {\mathfrak{o}}, \hat{\mathfrak{u}}_1, \hat{\mathfrak{u}}_2, \hat{\mathfrak{p}}]$. Here, the five parameters ${\mathfrak{i}}, {\mathfrak{o}}, \hat{\mathfrak{u}}_1, \hat{\mathfrak{u}}_2, \hat{\mathfrak{p}}$ are all polynomially bounded by the input length $l$ and $\frac{1}{\epsilon}$, so the number of table entries is $\mathcal{O}\left(\frac{l^3 n^2}{\epsilon^3}\right)$.

At each table entry, we compute the current value of the ratio of the welfare of agents based on the solution corresponding to that cell. Since this calculation involves only a constant number of arithmetic operations, the time per computation is at most $\mathcal{O}(l \log l)$. Accounting for both role symmetry between Agent~1 and Agent~2 and the two possible initial allocations in Phase~1, the total complexity increases by at most a factor of four.

We then consider only solutions with $\rho = \frac{W_1}{W_2} \geq 1$ and $\mathfrak{c} \leq B$. Trying all combinations of $i^\blacktriangleleft$, $i^\blacktriangleright$, and both directions requires $\mathcal{O}(n^2)$ time, and handling a possible {$u_1$-heavy resource} in Lemma~\ref{TRatio1} and~\ref{TRatio2} adds a factor of $n$. Thus, the overall complexity is $\mathcal{O}\left(\frac{n^8 l^4 \log l}{\epsilon^3}\right)$ and the knapsack problem has the complexity $\mathcal{O}\left(n^2\left\lfloor\frac{n}{\epsilon}\right\rfloor\right)$~\cite{vazirani2001approximation}.

Combining these observations, we obtain an FPTAS for AWNS-$\rho$.      
     \end{proof}

\begin{remark}
From the dynamic programming table, it is possible to find multiple solutions that minimize $\rho$. Among all these optimal solutions, the algorithm chooses to return those that are not Pareto dominated by any other optimal solution---that is, no other solution gives both agents at least as much welfare and one of them strictly more.

\end{remark}
This concludes the proof.
     \end{proof}
}

\section{Simulations}

This section presents a simulation study designed to evaluate the DSIRS framework on real-world data from the Spliddit platform~\cite{goldman2015spliddit}, where each instance specifies the allocation of 1{,}000 tokens by $n$ agents across $m$ resources. The first 5{,}000 Spliddit instances containing between 4 and 15 resources and at least two agents were selected for analysis. For each instance, two agents were sampled at random in order to generate a two-player subproblem. Utilities were restricted to the sampled agents for solving the AWNS-$\rho$ problem. In contrast, resource prices (representing revenues from sales) were computed using the full population of agents from the original utility matrix, whereas resource-specific selling costs were derived solely from the sampled pair. This reflects a setting in which market prices are determined by population-level valuations, while the cost of selling corresponds to the private emotional or logistical burden borne by the specific agents involved.

Each resource’s selling cost and sale price were instantiated item-wise via one of six mode pairs
\texttt{(avg, avg)}, \texttt{(max, max)}, \texttt{(avg, max)}, \texttt{(max, avg)}, \texttt{(max, min)}, and \texttt{(avg, min)}.
For each resource $j\in[m]$, we set $c_j=\mathrm{op}_c(u_1(j),u_2(j))$ from the sampled pair and $p_j^{\mathrm{raw}}=\mathrm{op}_p(\{u_i(j)\}_{i\in N})$ from the full population, where $\mathrm{op}_c,\mathrm{op}_p\in\{\mathrm{avg},\max,\min\}$ follow the chosen mode.
Population-level prices may exceed the disputants' valuations, i.e., $p_j^{\mathrm{raw}}>\max\{u_1(j),u_2(j)\}$; we retain this to model external market pricing and stress-test AWNS-$\rho$. For each instance and mode combination, we solved the AWNS-$\rho$ problem using our FPTAS, 
where $\varepsilon$ was fixed at $0.1$, ensuring that solutions remain close to the optimal 
while maintaining computational scalability. 
To mitigate order effects, the algorithm was executed once under each agent ordering, 
and the better outcome was retained. 
For resources that were heavily dominated by one agent in terms of utility, explicit 
forced-allocation and forced-sale variants were also evaluated, with the best outcome 
selected among all considered alternatives. 
The random seed was fixed at $42$ to ensure reproducibility of the sampling procedure. 
This simulation empirically evaluates the performance of our FPTAS for AWNS-$\rho$ and examines how effectively it achieves near-equitable outcomes under different cost/price 
assumptions and budget limits. The simulation evaluates the \emph{equitability of the resulting allocations}, measured by the welfare ratio $\rho$ and the difference $d$ between the two allocations. These measures reflect the proximity of the computed allocation to an equitable outcome under each cost/price mode and budget.

\label{sec:simulation}

\begin{figure*}[ht]
    \centering

    \begin{subfigure}{\linewidth}
        \centering
        \includegraphics[width=1.0\linewidth]{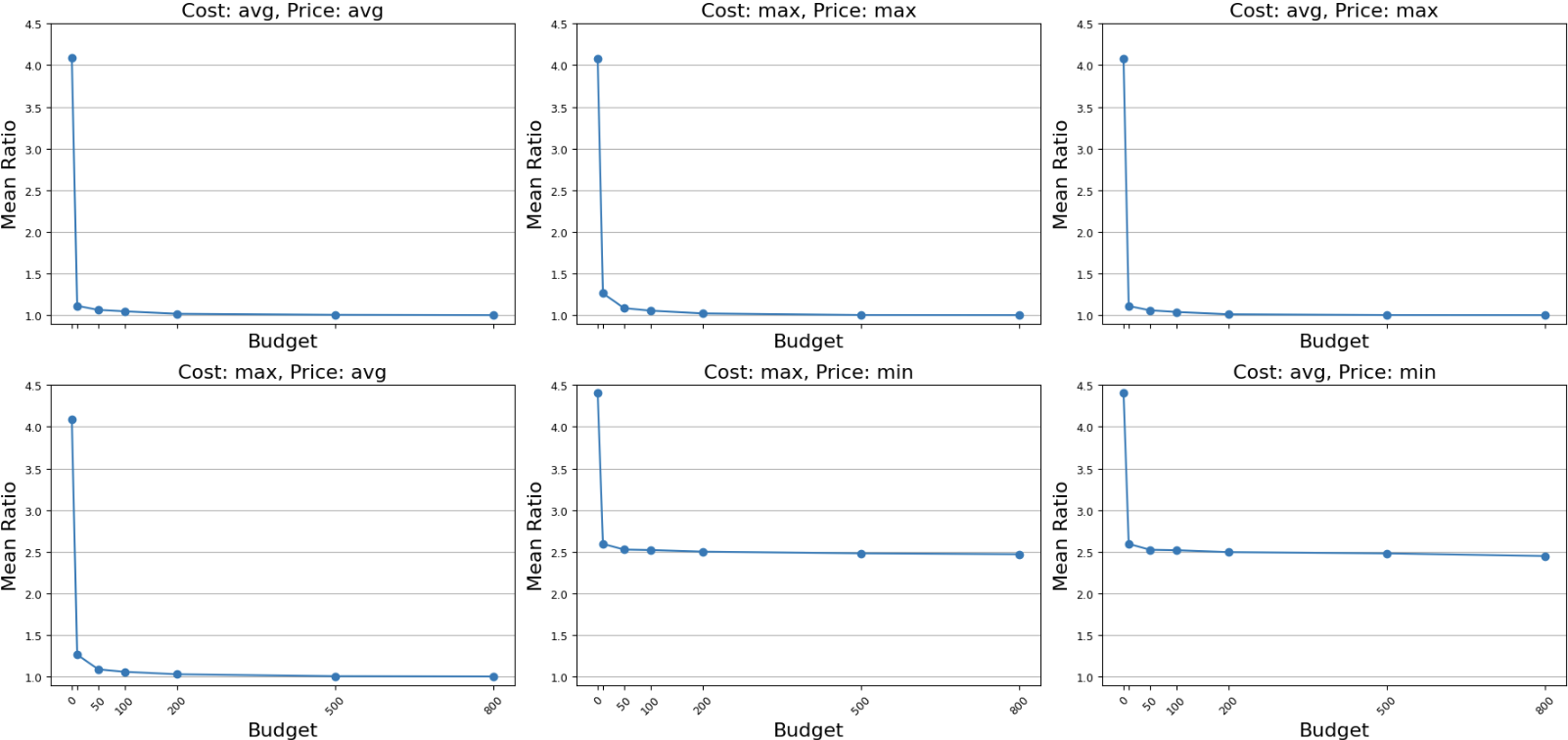}
        \caption{Mean $\rho=\max(\frac{W_1}{W_2},\frac{W_2}{W_1})$ as a function of the available budget, across six cost/price instantiations.}
        \label{fig:mean_ratio_vs_budget}
    \end{subfigure}

    \vspace{0.5cm}

    \begin{subfigure}{\linewidth}
        \centering
        \includegraphics[width=1.0\linewidth]{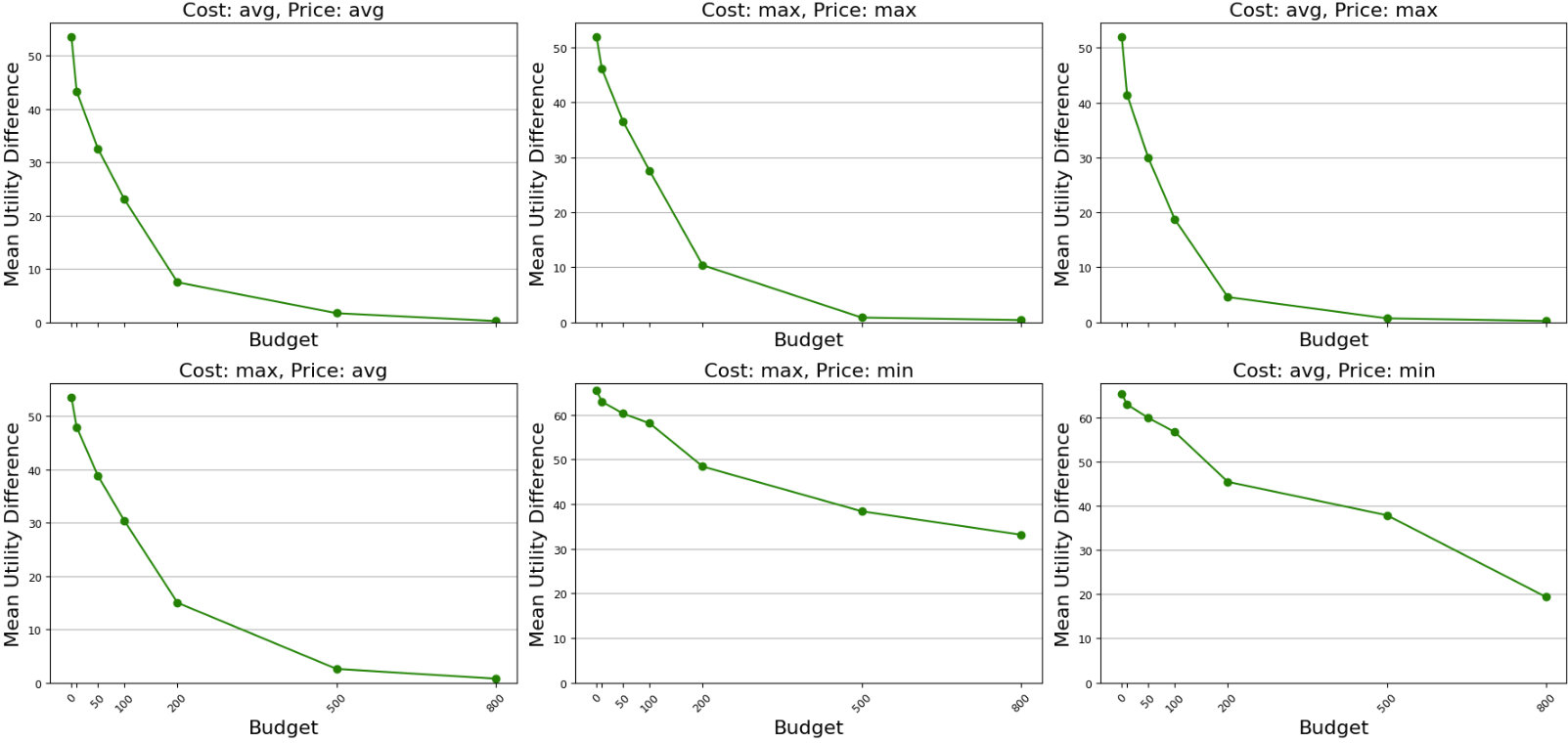}
        \caption{Mean $d=|W_1-W_2|$ as a function of the available budget, across six cost/price instantiations.}
        \label{fig:mean_difference_vs_budget}
    \end{subfigure}

\end{figure*}

Figure~\ref{fig:mean_ratio_vs_budget} reports the \emph{mean welfare ratio} $\rho$ for each mode as a function of the budget and Figure~\ref{fig:mean_difference_vs_budget} reports the \emph{mean absolute difference} between agent utilities $d$ for each mode as a function of the budget. The results reveal a pronounced difference in equitability across the modes. All modes exhibit sharp improvements in equitability as budget increases, particularly moving from no selling budget to any non-zero budget, where even modest budgets enable the algorithm to sell resources strategically and mitigate large imbalances. The most equitable outcomes are obtained under the average cost and average price
modes. In contrast, the two modes involving minimum price (average cost, minimum price and maximum cost and minimum price) yield significantly worse equitable outcomes. This effect is attributable to the limited redistribution power when prices are pessimistically estimated.

\section{Outlook}
In this paper, we addressed the fair allocation of indivisible resources between two agents within the DSIRS (Dispute Settlement with Indivisible Resources and Sale) framework. To avoid splitting, we proposed selling certain resources and redistributing the resulting revenue. We observed that the optimal fairness under different criteria may not be simultaneously achievable.

Building on the simplicity and transparency of the Adjusted Winner method, we formally defined a family of combinatorial problems that prohibit resource splitting. We showed that these problems are computationally intractable, with some variants even provably inapproximable. Nevertheless, we developed an FPTAS for the AWNS-$\rho$ problem and supported it with simulations, illustrating the impact of designing selling costs and revenue schemes.

Our findings offer both theoretical insights and practical tools for fair division through resource selling. We are particularly interested in exploring Pareto optimality in future work, that is, ensuring plans that are not Pareto-dominated---no agent’s welfare can be improved without making another worse off. Here, we point to several directions for future research:
\begin{enumerate}
    \item Analyzing structural properties of DSIRS instances, using theoretical tools and simulations (synthetically generated instances) of artificial data, inspired by the ``map of elections'' and ``putting fair division on the map''~\cite{szufa2020drawing,bohmputting}.
    \item Investigating whether selling “problematic” resources using each agent’s conversion rate between money and utility, representing how each agent internally values monetary compensation, followed by fair revenue redistribution, can reduce utility gaps while ensuring envy-freeness and Pareto optimality. While this variant is intractable, it is worth exploring whether our FPTAS can approximate utility ratios while preserving bounded EF.
    \item Studying a new fairness objective: maximizing the Nash product $W = W_1 \cdot W_2$, which encourages balance and guarantees Pareto optimality. This leads to the following variant:
\end{enumerate}
    \begin{definition}[\AWNSNW]
Given an instance of DSIRS~$\mathcal{I}$, find $S_0 \subseteq R$ such that the AW-derived plan $\mathcal{S} = P_{\text{AW}}(S_0, \mathcal{I})$ maximizes
\[
NW = W_1(\mathcal{S}, \mathcal{I}) \cdot W_2(\mathcal{S}, \mathcal{I}).
\]
\end{definition}






\bibliographystyle{unsrt}  
\bibliography{sampletrue}

@article{plaut2020almost,
  author       = {Benjamin Plaut and
                  Tim Roughgarden},
  title        = {Almost Envy-Freeness with General Valuations},
  journal      = {{SIAM} Journal on Discrete Mathematics
},
  volume       = {34},
  number       = {2},
  pages        = {1039--1068},
  year         = {2020}
}

@inproceedings{aleksandrovenvy,
  author       = {Martin Aleksandrov},
  title        = {Envy Freeness Up to One Item: Shall We Duplicate or Remove Resources?},
  booktitle   = {Progress in Artificial Intelligence - 21st Conference on Artificial Intelligence (EPIA'22) },
  series       = {Lecture Notes in Computer Science},
  volume       = {13566},
  pages        = {727--738},
  publisher    = {Springer},
  year         = {2022},
}

@inproceedings{amanatidis2022fair,
  author       = {Georgios Amanatidis and
                  Georgios Birmpas and
                  Aris Filos{-}Ratsikas and
                  Alexandros A. Voudouris},
  title        = {Fair Division of Indivisible Goods: {A} Survey},
  booktitle    = {Proceedings of the 31st International Joint Conference on Artificial Intelligence, (IJCAI' 22) },
  pages        = {5385--5393},
  year         = {2022},
publisher    = {ijcai.org},

}

@inproceedings{szufa2020drawing,
  author       = {Stanislaw Szufa and
                  Piotr Faliszewski and
                  Piotr Skowron and
                  Arkadii Slinko and
                  Nimrod Talmon},
  title        = {Drawing a Map of Elections in the Space of Statistical Cultures},
  booktitle    = {Proceedings of the 19th International Conference on Autonomous Agents and Multiagent Systems ({AAMAS} '20)},
  pages        = {1341--1349},
  year         = {2020}

}

@inproceedings{aziz2015adjusted,
  author       = {Haris Aziz and
                  Simina Br{\^{a}}nzei and
                  Aris Filos{-}Ratsikas and
                  S{\o}ren Kristoffer Stiil Frederiksen},
  title        = {The Adjusted Winner Procedure: Characterizations and Equilibria},
  booktitle    = {Proceedings of the 24th International Joint Conference on
                  Artificial Intelligence, (IJCAI '15)},
  pages        = {454--460},
  publisher    = {{AAAI} Press},
  year         = {2015},
}

@inproceedings{bohmputting,
  author       = {Paula B{\"{o}}hm and
                  Robert Bredereck and
                  Paul G{\"{o}}lz and
                  Andrzej Kaczmarczyk and
                  Stanislaw Szufa},
  title        = {Putting Fair Division on the Map},
  booktitle    = {Proceedings of the 40th {AAAI} Conference on Artificial Intelligence ({AAAI}'26)},
  publisher    = {{AAAI} Press},
  year         = {2026},
}

@book{brams1996fair,
  author       = {Steven J. Brams and
                  Alan D. Taylor},
  title        = {Fair division - from cake-cutting to dispute resolution},
  publisher    = {Cambridge University Press},
  year         = {1996},
}

@inproceedings{dupuis2009empirical,
  author       = {Nicolas Dupuis-Roy and Fr{\'{e}}d{\'{e}}ric Gosselin},
  title        = {An Empirical Evaluation of Fair-Division Algorithms},
  booktitle    = {Proceedings of the 31st Annual Meeting of the Cognitive Science Society},
  pages        = {31--36},
  publisher    = {Cognitive Science Society},
  address      = {Amsterdam, The Netherlands},
  year         = {2009},
  bibsource    = {Google Scholar}
}

@article{kilgour2024two,
  author       = {D. Marc Kilgour and Rudolf Vetschera},
  title        = {Two-Person Fair Division with Additive Valuations},
  journal      = {Group Decision and Negotiation},
  pages={745-774},
  volume={33},
  number={4},
  year={2024},
}

@incollection{pacuit2011towards,
  author       = {Eric Pacuit},
  editor       = {Johan van Benthem and
                  Amitabha Gupta and
                  Rohit Parikh},
  title        = {Towards a Logical Analysis of \emph{Adjusted Winner}},
  booktitle    = {Proof, Computation and Agency - Logic at the Crossroads},
  series       = {Synthese library},
  volume       = {352},
  pages        = {229--239},
  publisher    = {Springer},
  year         = {2011},

}

@article{goldman2015spliddit,
  title={Spliddit: Unleashing fair division algorithms},
  author={Goldman, Jonathan and Procaccia, Ariel D},
  journal={ACM SIGecom Exchanges},
  volume={13},
  number={2},
  pages={41--46},
  year={2015},
  publisher={ACM New York, NY, USA}
}

@inproceedings{bilo2024achieving,
  author       = {Vittorio Bil{\`{o}} and
                  Evangelos Markakis and
                  Cosimo Vinci},
  title        = {Achieving Envy-Freeness Through Items Sale},
  booktitle    = {32nd Annual European Symposium on Algorithms, {(ESA'24)}},
  series       = {LIPIcs},
  volume       = {308},
  pages        = {26:1--26:16},
  year         = {2024},
}

@inproceedings{aziz2016control,
  author       = {Haris Aziz and
                  Ildik{\'{o}} Schlotter and
                  Toby Walsh},
  title        = {Control of Fair Division},
  booktitle    = {Proceedings of the 25th International Joint Conference on
                  Artificial Intelligence, (IJCAI'16)},
  pages        = {67--73},
  publisher    = {{IJCAI/AAAI} Press},
  year         = {2016},
}

@article{folberg2009mediating,
  author       = {Jay Folberg},
  title        = {Mediating Family Property and Estate Conflicts},
  journal = {JAMS Dispute Resolution ALERT},
  year    = {2009},
  volume  = {9},
}

@article{aziz2022algorithmic,
  author       = {Haris Aziz and
                  Bo Li and
                  Herv{\'{e}} Moulin and
                  Xiaowei Wu},
  title        = {Algorithmic fair allocation of indivisible items: a survey and new
                  questions},
  journal      = {SIGecom Exch.},
  volume       = {20},
  number       = {1},
  pages        = {24--40},
  year         = {2022},
}

@article{massoud2000fair,
  author       = {Tansa George Massoud},
  title        = {Fair Division, Adjusted Winner Procedure (AW), and the Israeli-Palestinian Conflict},
  journal      = {Journal of Conflict Resolution},
  volume       = {44},
  number       = {3},
  pages        = {333--358},
  year         = {2000},
  publisher    = {Sage Publications, Inc., 2455 Teller Road, Thousand Oaks, CA 91320},
  bibsource    = {Google Scholar}
}

@article{brams1997fair,
  author       = {Steven J. Brams},
  title        = {Fair Division: A New Approach to the Spratly Islands Controversy},
  journal      = {International Negotiation},
  volume       = {2},
  number       = {2},
  pages        = {303--329},
  year         = {1997},
  publisher    = {Brill | Nijhoff, The Netherlands},
  bibsource    = {Google Scholar}
}

@inproceedings{bouveret2016characterizing,
  author       = {Sylvain Bouveret and
                  Michel Lema{\^{\i}}tre},
  title        = {Characterizing conflicts in fair division of indivisible goods using a scale of criteria},
  journal      = {Proceedings of the 15th International Conference on Autonomous Agents and MultiAgent Systems (AAMAS'16)},
  volume       = {30},
  number       = {2},
  pages        = {259--290},
  year         = {2016},
}

@incollection{faliszewski2016control,
  author       = {Piotr Faliszewski and
                  J{\"{o}}rg Rothe},
  editor       = {Felix Brandt and
                  Vincent Conitzer and
                  Ulle Endriss and
                  J{\'{e}}r{\^{o}}me Lang and
                  Ariel D. Procaccia},
  title        = {Control and Bribery in Voting},
  booktitle    = {Handbook of Computational Social Choice},
  pages        = {146--168},
  publisher    = {Cambridge University Press},
  year         = {2016},
}

@article{bartholdi1992hard,
  author       = {John J. Bartholdi III and Craig A. Tovey and Michael A. Trick},
  title        = {How Hard Is It to Control an Election?},
  journal      = {Mathematical and Computer Modelling},
  volume       = {16},
  number       = {8--9},
  pages        = {27--40},
  year         = {1992},
  publisher    = {Elsevier},
}

@article{salame2005some,
  author       = {Samer Salame and Eric Pacuit and Rohit Parikh},
  title        = {Some Results on Adjusted Winner},
  journal      = {Synthese},
  volume       = {142},
  number       = {1},
  pages        = {1--30},
  year         = {2005},
  publisher    = {Springer},
  bibsource    = {Google Scholar},
}

@book{garey1979computers,
  author       = {Michael R. Garey and David S. Johnson},
  title        = {Computers and Intractability: A Guide to the Theory of {NP}-Completeness},
  year         = {1979},
  publisher    = {W. H. Freeman and Company},
}

@book{vazirani2001approximation,
  author       = {Vijay V. Vazirani},
  title        = {Approximation Algorithms},
  year         = {2001},
  publisher    = {Springer}
}

@book{brandt2016handbook,
  editor       = {Felix Brandt and
                  Vincent Conitzer and
                  Ulle Endriss and
                  J{\'{e}}r{\^{o}}me Lang and
                  Ariel D. Procaccia},
  title        = {Handbook of Computational Social Choice},
  publisher    = {Cambridge University Press},
  year         = {2016},
}

@inproceedings{HalpernS19,
  author       = {Daniel Halpern and
                  Nisarg Shah},
  title        = {Fair Division with Subsidy},
  booktitle    = {Proceedings of 12th International Symposium, (SAGT'19)
                 },
  series       = {Lecture Notes in Computer Science},
  volume       = {11801},
  pages        = {374--389},
  publisher    = {Springer},
  year         = {2019},
}

@article{ChakrabortyISZ21,
  author       = {Mithun Chakraborty and
                  Ayumi Igarashi and
                  Warut Suksompong and
                  Yair Zick},
  title        = {Weighted Envy-freeness in Indivisible Item Allocation},
  journal      = {{ACM} Transactions on Economics and Computation},
  volume       = {9},
  number       = {3},
  pages        = {18:1--18:39},
  year         = {2021},
}

@inproceedings{BrustleDNSV20,
  author       = {Johannes Brustle and
                  Jack Dippel and
                  Vishnu V. Narayan and
                  Mashbat Suzuki and
                  Adrian Vetta},
  title        = {One Dollar Each Eliminates Envy},
  booktitle    = {Proceedings of the 21st {ACM} Conference on Economics and Computation (EC'20)},
  pages        = {23--39},
  publisher    = {{ACM}},
  year         = {2020},
}

@article{ChaudhuryKMS21,
  author       = {Bhaskar Ray Chaudhury and
                  Telikepalli Kavitha and
                  Kurt Mehlhorn and
                  Alkmini Sgouritsa},
  title        = {A Little Charity Guarantees Almost Envy-Freeness},
  journal      = {{SIAM} J. Comput.},
  volume       = {50},
  number       = {4},
  pages        = {1336--1358},
  year         = {2021},
}

@inproceedings{CaragiannisI21,
  author       = {Ioannis Caragiannis and
                  Stavros Ioannidis},
  title        = {Computing Envy-Freeable Allocations with Limited Subsidies},
  booktitle    = {Proceedings of 17th International Conference on Web and Internet Economics  (WINE '21)},
  series       = {Lecture Notes in Computer Science},
  volume       = {13112},
  pages        = {522--539},
  publisher    = {Springer},
  year         = {2021},
}



\end{document}